\documentclass[runningheads,a4paper]{llncs}
\usepackage{cmap}
\usepackage{xcolor}
\usepackage{amsmath}
\usepackage{comment}
\usepackage{amssymb}
\usepackage{paralist}
\usepackage{wrapfig,pgfplots,tikz}
\pgfplotsset{compat=newest}
\usetikzlibrary{automata,arrows,patterns}
\usepackage[utf8]{inputenc}
\usepackage[english]{babel}
\usepackage{algpascal,algorithm}
\usepackage[noadjust]{cite}

\usepackage{xcolor}

\newcommand{\set}[1]{\{#1\}}
\newcommand{\setcond}[2]{\{#1\mid\,#2\}}

\newcommand{\projection}[2]{#1\!\downarrow\!#2}

\newcommand{\mytrue}{\ensuremath{\mathit{true}}}
\newcommand{\myfalse}{\ensuremath{\mathit{false}}}

\newcommand{\pspace}{\textsc{PSpace}}

\newcommand{\disunion}{\uplus}

\newcommand{\bottom}{\bot}

\newcommand{\nat}{\mathbb{N}}
\newcommand{\intrange}[2]{[#1..#2]}

\newcommand{\oracle}{\mathcal{O}}
\newcommand{\oracleOf}[1]{\oracle(#1)}

\newcommand{\tso}{\text{TSO}}%
\newcommand{\seqcon}{\text{SC}}%
\newcommand{\memmodel}{\text{M}}

\newcommand{\trencher}{\textsc{Trencher}} %
\newcommand{\memorax}{\textsc{Memorax}}
\newcommand{\cbmc}{\textsc{CBMC}}
\newcommand{\spin}{\textsc{SPIN}}

\newcommand{\automaton}{\ensuremath{A}}
\newcommand{\alphabet}{\ensuremath{\Sigma}}
\newcommand{\states}{\ensuremath{S}}
\newcommand{\transitions}{\ensuremath{\rightarrow}}
\newcommand{\astate}{\ensuremath{s}}
\newcommand{\initialstate}{\ensuremath{\astate_0}}
\newcommand{\finalstates}{\ensuremath{F}}
\newcommand{\transitionto}[1]{\ensuremath{\xrightarrow{#1}}}

\newcommand{\langFinOf}[2]{\ensuremath{\mathcal{L}_{#1}({#2})}}
\newcommand{\succorder}[1]{\ensuremath{<_{#1}}}

\newcommand{\aprogram}{\ensuremath{\mathit{P}}}
\newcommand{\commands}{\ensuremath{\mathit{Com}}}
\newcommand{\controlstates}{\ensuremath{Q}}
\newcommand{\instructions}{\ensuremath{\mathit{I}}}
\newcommand{\acontrolstate}{\ensuremath{q}}
\newcommand{\acontrolstateOf}[1]{\acontrolstate_{#1}}

\newcommand{\initialcontrolstateOf}[1]{\acontrolstate_{0,#1}}
\newcommand{\finalcontrolstateOf}[1]{\acontrolstate_{g,#1}}
\newcommand{\commandOf}[1]{\mathit{cmd}(#1)}
\newcommand{\sourcestateOf}[1]{\ensuremath{\mathit{src}(#1)}}
\newcommand{\destinationstateOf}[1]{\ensuremath{\mathit{dst}(#1)}}

\newcommand{\threaddomain}{\textsf{TID}} %
\newcommand{\datadomain}{\textsf{DOM}}
\newcommand{\addrdomain}{\textsf{DOM}}
\newcommand{\regdomain}{\textsf{REG}}
\newcommand{\expdomain}{\textsf{EXP}}
\newcommand{\areg}{r}
\newcommand{\anaddr}{a}
\newcommand{\aval}{v}
\newcommand{\anexpr}{e}
\newcommand{\theload}[2]{\ensuremath{{#1}\leftarrow \mathtt{mem[}{#2}\mathtt{]}}}
\newcommand{\thestore}[2]{\ensuremath{\mathtt{mem[}{#1}\mathtt{]}\leftarrow{#2}}}
\newcommand{\thelocal}[2]{\ensuremath{{#1}\leftarrow{#2}}}
\newcommand{\theassume}[1]{\ensuremath{\mathtt{assume~}{#1}}}
\newcommand{\themfence}[0]{\ensuremath{\mathtt{mfence}}}

\newcommand{\therules}[1]{{\small\setlength{\arraycolsep}{1em}$$\begin{array}{c}#1\end{array}$$}}
\newcommand{\therule}[3]{\dfrac{\text{#1}}{\text{#2}}~(\text{#3})\\[2em]}
\newcommand{\doublerule}[6]{\dfrac{\text{#1}}{\text{#2}}~(\text{#5})\qquad\dfrac{\text{#3}}{\text{#4}}~(\text{#6})\\[2em]}

\newcommand{\eventautomaton}[2]{\ensuremath{X_{#1}(#2)}}
\newcommand{\events}{\ensuremath{E}}
\newcommand{\eventstatesOf}[1]{\states_{#1}}
\newcommand{\eventtransitions}[1]{\Delta_{#1}}
\newcommand{\initialeventstate}{\astate_{0}}
\newcommand{\finaleventstates}{\finalstates}

\newcommand{\indexconf}{\mathsf{ec}}
\newcommand{\pcconf}{\mathsf{pc}}
\newcommand{\valconf}{\mathsf{val}}
\newcommand{\bufconf}{\mathsf{buf}}

\newcommand{\initpcconf}{\pcconf_0}
\newcommand{\initvalconf}{\valconf_0}
\newcommand{\initbufconf}{\bufconf_0}
\newcommand{\writekind}{\ensuremath{\text{flush}}}

\newcommand{\storeevent}{\ensuremath{\mathtt{store}}}
\newcommand{\loadevent}{\ensuremath{\mathtt{load}}}
\newcommand{\flushevent}{\ensuremath{\mathtt{flush}}}
 
\newcommand{\assignevent}{\ensuremath{\mathtt{assign}}}
\newcommand{\fenceevent}{\ensuremath{\mathtt{fence}}}

\newcommand{\anevent}{\mathtt{e}} %

\newcommand{\threadOf}[1]{\mathit{thread}(#1)}
\newcommand{\instructionOf}[1]{\mathit{inst}(#1)}
\newcommand{\addrOf}[1]{\mathit{addr}(#1)}

\newcommand{\anindex}{\mathit{id}}
\newcommand{\aninstruction}{\ensuremath{\mathit{inst}}}
\newcommand{\acommand}{\ensuremath{\mathit{cmd}}}
\newcommand{\valOf}[1]{\ensuremath{\widehat{#1}}}
\newcommand{\computationsOf}[2]{\mathcal{C}_{#1}({#2})} 
\newcommand{\reachOf}[2]{\textit{Reach}_{#1}{(#2)}}
\newcommand{\goalstates}{G}
\newcommand{\acceptingstates}{\ensuremath{\goalstates}} %

\newcommand{\one}[1]{{#1}_1}
\newcommand{\two}[1]{{#1}_2}
\newcommand{\xaddr}{\text{x}}
\newcommand{\yaddr}{\text{y}}

\newcommand{\wit}{\text{wit}}

\newcommand{\hb}{\mathit{hb}}
\newcommand{\po}{\mathit{po}}
\newcommand{\cf}{\mathit{cf}}
\newcommand{\equivalence}{}

\newcommand{\order}[1]{\rightarrow_{#1}}
\newcommand{\progorder}{\order{\po}}
\newcommand{\conflictorder}{\order{\cf}}
\newcommand{\happensbefore}{\order{\hb}}
\newcommand{\happensbeforeOf}[1]{\order{\hb}(#1)}
\newcommand{\equivalenceorder}{\leftrightarrow_{\equivalence}}

\newcommand{\attackthread}{\athread}%
\newcommand{\attackstore}{\storeevent}%
\newcommand{\attackflush}{\flushevent}%
\newcommand{\attackload}{\loadevent}%

\newcommand{\athread}{t}
\newcommand{\aprog}{\ensuremath{\mathit{P}}}
\newcommand{\bprog}{\ensuremath{\mathit{R}}}

\newcommand{\atsocomput}{\tau}

\newcommand{\awitness}{\atsocomput}
\newcommand{\glue}{\ensuremath{\oplus}}
\newcommand{\reowit}{\sigma}
\newcommand{\reowitof}[1]{\reowit_{#1}}

\newcommand{\wita}{\text{(W1)}}
\newcommand{\witb}{\text{(W2)}}
\newcommand{\witc}{\text{(W3)}}
\newcommand{\witd}{\text{(W4)}}
\newcommand{\wite}{\text{(W5)}}

\newcommand{\adrreg}{ar}
\newcommand{\valreg}{vr}
\newcommand{\bound}{\mathtt{max}}
\newcommand{\iterator}{\mathtt{count}}

\newcommand{\atransition}[3]{(\ensuremath{#1},\,#3,\,\ensuremath{#2})}
\newcommand{\wsproj}[1]{\ensuremath{h_{#1}}}
\newcommand{\proj}[1]{\ensuremath{f_{#1}}}

\newcommand{\statesextended}{\states_\glue}
\newcommand{\eventsextended}{\events_\glue}
\newcommand{\initialextended}{\astate_\glue}
\newcommand{\finalextended}{\finalstates_\glue}

\newcommand{\instructionsextended}{\instructions_\glue}

\newcommand{\acomp}{\ensuremath{\alpha}}
\newcommand{\bcomp}{\ensuremath{\beta}}

\newcommand{\extensionstate}[1]{\ensuremath{{\overline\acontrolstate}_{#1}}}

\usepackage{etex,etoolbox}

\makeatletter
\providecommand{\@fourthoffour}[4]{#4}
\def\fixstatement#1{%
  \AtEndEnvironment{#1}{%
    \xdef\pat@label{\expandafter\expandafter\expandafter
      \@fourthoffour\csname#1\endcsname\space\@currentlabel}}}

\globtoksblk\prooftoks{1000}
\newcounter{proofcount}

\long\def\proofatend#1\endproofatend{%
  \edef\next{\noexpand\begin{proof}[of {\#}\pat@label]}%
  \toks\numexpr\prooftoks+\value{proofcount}\relax=\expandafter{\next#1\end{proof}}
  \stepcounter{proofcount}}

\long\def\lemmaatend#1\endlemmaatend{%
  \edef\next{}%
  \toks\numexpr\prooftoks+\value{proofcount}\relax=\expandafter{\next#1}
  \stepcounter{proofcount}}

\def\printproofs{%
  \count@=\z@
  \loop
    \the\toks\numexpr\prooftoks+\count@\relax
    \ifnum\count@<\value{proofcount}%
    \advance\count@\@ne
  \repeat}
\makeatother

\newtheorem{atheorem}{Theorem}
\newtheorem{alemma}[atheorem]{Lemma}
\fixstatement{atheorem}
\fixstatement{alemma}

\newcommand{\acks}{\paragraph*{Acknowledgements}}

\usepackage{setspace}
\begin{document}
\title{Lazy \tso{} Reachability}
\author{Ahmed Bouajjani\inst{1} \and Georgel Calin\inst{2} \and Egor Derevenetc\inst{2,3} \and Roland Meyer\inst{2}}
\authorrunning{A. Bouajjani, G. Calin, E. Derevenetc, R. Meyer}
\institute{\hspace{-0.5em}$^1$LIAFA, University Paris 7 \quad $^2$University of Kaiserslautern \quad $^3$Fraunhofer ITWM}
\newif\ifappendix
\maketitle
\begin{abstract}
We address the problem of checking state reachability for programs running under Total Store Order (\tso). 
The problem has been shown to be decidable but the cost is prohibitive, namely non-primitive recursive.
We propose here to give up completeness.
Our contribution is a new algorithm for TSO reachability: 
it uses the standard SC semantics and introduces the TSO semantics lazily and only where needed.
At the heart of our algorithm is an iterative refinement of the program of interest.
If the program's goal state is SC-reachable, we are done.
If the goal state is not SC-reachable, this may be due to the fact that SC under-approximates TSO.
We employ a second algorithm that determines TSO computations which are infeasible under SC, and hence likely to lead to new states.
We enrich the program to emulate, under SC, these TSO computations. %
Altogether, this yields an iterative under-approximation that we prove sound and complete for bug hunting, i.e., a semi-decision procedure halting for positive cases of reachability.
We have implemented
the procedure as an extension to the tool Trencher~\cite{Trencher} and compared
it to the \memorax{}~\cite{AbdullaACLR12} and \cbmc{}~\cite{Clarke04atool} model checkers.
\end{abstract}
\section{Introduction}
Sequential consistency (\seqcon)~\cite{Lamport79} is the semantics typically  assumed for parallel programs.
Under \seqcon{}, instructions are executed atomically and in program order.
When programs are executed on an Intel x86 processor, however, they are only guaranteed a weaker semantics known as Total Store Order (\tso).
\tso\ weakens the synchronization guarantees given by \seqcon{}, which in turn may lead to erroneous behavior.
\tso\ reflects the architectural optimization of store buffers.
To reduce the latency of memory accesses, store commands are added to a thread-local FIFO buffer and only later executed on memory.

To check for correct behavior, reachability techniques have proven useful.
Given a program and a goal state, the task is to check whether the state is reachable.
To give an example, assertion failures can be phrased as reachability problems.
Reachability depends on the underlying semantics.
Under \seqcon, the problem is known to be \pspace-complete~\cite{Kozen77}.
Under \tso{}, it is considerably more difficult: although decidable, it is non-primitive recursive-hard~\cite{ABBM10}.

Due to the high complexity,
tools rarely provide decision procedures~\cite{AbdullaACLR12,LW10,LindenW11}.
Instead, most approaches implement approximations.
Typical approximations of \tso\ reachability bound the number of loop iterations~\cite{AlglaveKNT13,PartialOrders},
the number of context switches between threads~\cite{ABP11},
or the size of store buffers~\cite{KVY2011,KupersteinVY12}.
What all these approaches have in common is that they introduce store buffering in the \emph{whole} program.
We claim that such a comprehensive instrumentation is unnecessarily heavy.

The idea of our method is to introduce store buffering lazily and only where needed.
Unlike~\cite{AbdullaACLR12}, we do not target completeness.
Instead, we argue that our lazy \tso{} reachability checker is useful for a fast detection of bugs that are due to the \tso{} semantics.
At a high level, we solve the expensive \tso{} reachability problem with a series of cheap \seqcon{} reachability checks --- very much like SAT solvers are invoked as subroutines of costlier analyses.
The \seqcon\ checks run interleaved with queries to an oracle.
The task of the oracle is to suggest sequences of instructions that
should be considered under \tso, which means they are likely to lead to \tso-reachable states outside \seqcon.

To be more precise, the algorithm iteratively repeats the following steps.
First, it checks whether the goal state is \seqcon{}-reachable.
If this is the case, the state will be \tso{}-reachable as well and the algorithm returns.
If the state is not \seqcon-reachable, the algorithm asks the oracle for a sequence of instructions and encodes the \tso{} behavior of the sequence into the input program.
As a result, precisely this \tso{} behavior becomes available under \seqcon{}.
The encoding is linear in the size of the input program and in the length of the sequence.

The algorithm is a semi-decision procedure: it always returns correct answers and is guaranteed to terminate if the goal state is \tso{}-reachable.
This guarantee relies on one assumption on the oracle.
If the oracle returns the empty sequence, then the \seqcon{}- and the \tso{}-reachable states of the input program have to coincide.
We also come up with a good oracle:
robustness checkers naturally meet the above requirement.
Intuitively, a program is robust against \tso{} if its partial order-behaviors (reflecting data and control dependencies) under \tso\ and under \seqcon\ coincide.
Robustness is much easier than \tso{} reachability, actually \pspace{}-complete~\cite{BMM11,BDM13}, and hence well-suited for iterative invocations.

We have im\-ple\-men\-ted lazy TSO reachability as an ex\-ten\-sion to our tool \trencher~\cite{Trencher}, reusing the robustness checking algorithms of \trencher{} to derive an oracle.
The implementation is able to solve positive instances of \tso{} reachability as well as correctly determine safety for robust programs.
The source code and experiments are available online~\cite{Trencher}. %

The structure of the paper is as follows.
We introduce parallel programs with their \tso\ and their \seqcon\ semantics in Section~\ref{Section:Programs}.
Section~\ref{Section:LazyReachability} presents our main contribution, the lazy approach to solving \tso{} reachability.
Section~\ref{Section:Robustness} describes the robustness-based oracle.
The experimental evaluation is given in Section~\ref{Section:Experiments}.
Details and proofs missing in the main text can be found in the appendix.

\subsection*{Related Work}
As already mentioned, \tso{} reachability was proven decidable but non-primitive recursive~\cite{ABBM10} in the case of a finite number of threads and a finite data domain.
In the same setting, robustness was shown to be \pspace-complete~\cite{BMM11}.
Checking and enforcing robustness against weak memory models has been addressed in
\cite{burckhardt-musuvathi-CAV08,Sen2011,Alglave2010,AlglaveM11,BMM11,BDM13,ShashaSnir88}.
The first work to give an efficient sound and complete decision procedure for checking robustness is~\cite{BDM13}.

The works~\cite{AbdullaACLR12,LW10,LindenW11} propose state-based techniques to solve \tso{} reachability.
An under-approximative method that uses bounded context switching is given in~\cite{ABP11}.
It encodes store buffers into a linear-size instrumentation, and the instrumented program is checked for \seqcon{} reachability.
The under-approximative techniques of~\cite{AlglaveKNT13,PartialOrders} are able to guarantee safety only for programs with bounded loops.
On the other side of the spectrum,
over-approximative analyses abstract store buffers into sets combined with bounded queues~\cite{KVY2011,KupersteinVY12}. 
\section{Parallel Programs}\label{Section:Programs} %
We use automata to define the syntax and the semantics of parallel programs.
A (non-deterministic) \emph{automaton} over an alphabet $\alphabet$ is a tuple $\automaton=(\alphabet,\states,\transitions,\initialstate)$,
where $\states$ is a set of states,
$\transitions\,\subseteq\states\times(\alphabet\cup\set{\varepsilon})\times\states$ is a set of transitions, and
$\initialstate\in\states$ is an initial state.
The automaton is \emph{finite} if the transition relation $\transitions$ is finite. %
We write $\astate\transitionto{a}\astate'$ if $(\astate,a,\astate')\in\,\transitions$,
and extend the transition relation to sequences $w\in\alphabet^*$ as expected.
The \emph{language of $\automaton$ with final states $\finalstates\subseteq \states$} is
$\langFinOf{\finalstates}{\automaton}:=\setcond{w\in\alphabet^*}{\initialstate\transitionto{w}\astate\in\finalstates}$.
We say that state $\astate\in\states$ is \emph{reachable} if $\initialstate\transitionto{w}\astate$ for some sequence $w\in\alphabet^*$.
Letter \emph{$a$ precedes $b$ in $w$}, denoted by $a\succorder{w}b$,
if $w=w_1\cdot a\cdot w_2\cdot b\cdot w_3$ for some $w_1,w_2,w_3\in\alphabet^*$.

A parallel program $\aprogram$ is a finite sequence of threads that are identified by indices $\athread$ from $\threaddomain$.
Each thread $\athread:=
(\commands_\athread,\controlstates_\athread,\instructions_\athread,\initialcontrolstateOf{\athread})$
is a finite automaton with transitions $\instructions_\athread$ that we call \emph{instructions}.
Instructions $\instructions_\athread$ are labelled by \emph{commands} from the set $\commands_\athread$ which we define in the next paragraph.
We assume, wlog., that states of different threads are disjoint.
This implies that the sets of instructions of different threads are distinct.
We use $\instructions:=\biguplus_{\athread\in\threaddomain}\instructions_\athread$ for all instructions and
$\commands:=\bigcup_{\athread\in\threaddomain}\commands_\athread$ for all commands.
For an instruction $\aninstruction:=(\astate,\acommand,\astate')$ in $\instructions$, we define $\commandOf{\aninstruction}:=\acommand$, $\sourcestateOf{\aninstruction}:=\astate$, and $\destinationstateOf{\aninstruction}:=\astate'$.

\begin{wrapfigure}{r}{0.465\textwidth}
	\vskip -2.5em
	\centering\small
	\begin{tikzpicture}[->,>=stealth',shorten >=1pt, auto, node distance=1cm, transform shape, scale=1,baseline=.65cm]%
		\tikzstyle{every state}=[fill=white,circle,draw=black,inner sep=0pt,text=black,minimum size=6pt]
		\begin{scope} %
			\node[state,initial,initial left,initial text=$\athread_1$,label=right:$\acontrolstate_{0,1}$] (q10) {};
			\node[state,label=right:$\acontrolstate_{1,1}$] (q11) [below of=q10] {};
			\node[state,label=right:$\acontrolstate_{2,1}$] (q12) [below of=q11] {};
			\node[accepting,state,label=right:$\finalcontrolstateOf{1}$] (q13) [below of=q12] {};
			\path	(q10)		edge	node [swap] {$\thestore{\xaddr}{1}$} (q11)
					(q11)		edge 	node [swap] {$\theload{r_1}{\yaddr}$} (q12)
					(q12)		edge	node [swap] {$\theassume{r_1\hspace{-.2em}=\hspace{-.1em}0}$} (q13);
		\end{scope}
		\begin{scope} %
			\node[state,initial,initial left,initial text=$\athread_2$,label=right:$\acontrolstate_{0,2}$] (q20) [node distance=2.65cm,right of=q10] {};
			\node[state,label=right:$\acontrolstate_{1,2}$] (q21) [below of=q20] {};
			\node[state,label=right:$\acontrolstate_{2,2}$] (q22) [below of=q21] {};
			\node[accepting,state,label=right:$\finalcontrolstateOf{2}$] (q23) [below of=q22] {};
			\path	(q20)		edge	node [swap] {$\thestore{\yaddr}{1}$} (q21)
					(q21)		edge 	node [swap] {$\theload{r_2}{\xaddr}$} (q22)
					(q22)		edge	node [swap] {$\theassume{r_2\hspace{-.2em}=\hspace{-.1em}0}$} (q23);
		\end{scope}	
	\end{tikzpicture}
	\caption{Simplified Dekker's algorithm.}
	\label{Figure:Dekker}
	\vskip -1em
\end{wrapfigure}
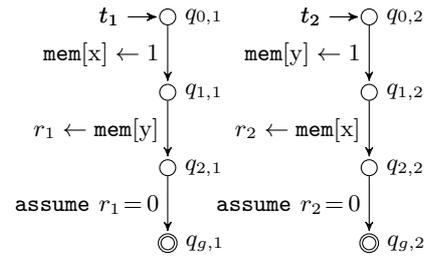

To define the set of commands, let $\datadomain$ be a finite domain of values that we also use as addresses.
We assume that value $0$ is in $\datadomain$.
For each thread $\athread$, let $\regdomain_\athread$ be a finite set of registers that take their values from $\datadomain$.
We assume per-thread disjoint sets of registers.
The set of expressions of thread $\athread$, denoted by $\expdomain_\athread$, is defined over registers from $\regdomain_\athread$,
constants from $\datadomain$, and (unspecified) operators over $\datadomain$.
If $\areg\in\regdomain_\athread$ and $\anexpr,\anexpr'\in\expdomain_\athread$,
the set of commands $\commands_\athread$ consists of
loads from memory $\theload{\areg}{\anexpr}$,
stores to memory $\thestore{\anexpr}{\anexpr'}$,
memory fences $\themfence$,
assignments $\thelocal{\areg}{\anexpr}$,
and conditionals $\theassume{\anexpr}$.
We write $\regdomain:=\biguplus_{\athread\in\threaddomain}\regdomain_\athread$ for all registers
and $\expdomain:=\bigcup_{\athread\in\threaddomain}\expdomain_\athread$ for all expressions.

The program in Figure~\ref{Figure:Dekker} serves as our running example.
It consists of two threads $\athread_1$ and $\athread_2$ implementing a mutual exclusion protocol.
Initially, the addresses $\xaddr$ and $\yaddr$ contain $0$.
The first thread signals its intent to enter the critical section by setting variable $\xaddr$ to $1$.
Next, the thread checks whether the second thread wants to enter the critical section, too.
It reads variable $\yaddr$ and, if it is $0$, the first thread enters its critical section.
The critical section actually is the state $\finalcontrolstateOf{1}$.
The second thread behaves symmetrically.
\subsection{Semantics of Parallel Programs}\label{Subsection:Semantics}
The semantics of a parallel program $\aprogram$ under memory model $\memmodel=\tso$ and $\memmodel=\seqcon$ follows~\cite{Owens2009}.
We define the semantics in terms of a \emph{state-space automaton}
$\eventautomaton{\memmodel}{\aprogram}:=(\events,\eventstatesOf{\memmodel},\eventtransitions{\memmodel},\initialeventstate)$.
Each state $\astate=(\pcconf, \valconf, \bufconf)\in\eventstatesOf{\memmodel}$ is a tuple where
the program counter $\pcconf\colon\threaddomain\rightarrow\controlstates$ holds the current control state of each thread,
the valuation $\valconf\colon\regdomain\,\cup\,\addrdomain\rightarrow\datadomain$ holds the values stored in registers and at memory addresses,
and the buffer configuration $\bufconf\colon\threaddomain\rightarrow(\addrdomain\times\datadomain)^*$
holds a sequence of address-value pairs.

In the \emph{initial state} $\initialeventstate := (\initpcconf, \initvalconf, \initbufconf)$,
the program counter holds the initial control states, $\initpcconf(\athread):=\initialcontrolstateOf{\athread}$ for all $\athread\in\threaddomain$,
all registers and addresses contain value $0$,
and all buffers are empty, $\initbufconf(\athread):=\varepsilon$ for all $t\in\threaddomain$.

The transition relation $\eventtransitions{\tso}$ for $\tso$ satisfies the rules given in Figure~\ref{Figure:TSORules}.
There are two more rules for register assignments and conditionals that are standard and omitted.
TSO architectures implement (FIFO) store buffering, which means stores are buffered for later execution on the shared memory.
Loads from an address $\anaddr$ take their value from the most recent store to address $\anaddr$ that is buffered.
If there is no such buffered store, they access the main memory.
This is modelled by the Rules~({LB}) and~({LM}).
Rule~({ST}) enqueues store operations as address-value pairs to the buffer.
Rule~({MEM}) non-deterministically dequeues store operations and executes them on memory.
Rule~({F}) states that a thread can execute a fence only if its buffer is empty.
As can be seen from Figure~\ref{Figure:TSORules}, events labelling TSO transitions take the form
$\events\subseteq\threaddomain\times(\instructions\cup\set{\writekind})\times(\addrdomain\cup\set{\bot})$.

\begin{figure}[!t]
\centering
\therules{
	\hspace{-1em}
	\therule{ %
		\acommand\,= \theload{\areg}{\anexpr_\anaddr}
		\quad$\projection{\bufconf(\athread)}{(\set{\anaddr}\times\datadomain)} = ({\anaddr},{\aval})\cdot\beta$
	}{$\astate\transitionto{(\athread,\aninstruction,\anaddr)} (\pcconf',\valconf[\areg:=\aval],\bufconf)$}
	{LB}
	\hspace{-1em}
	\therule{ %
		\acommand\,= \theload{\areg}{\anexpr_\anaddr} 
		\quad$\projection{\bufconf(\athread)}{(\set{\anaddr}\times\datadomain)} = \varepsilon$
	}{$\astate\transitionto{(\athread,\aninstruction,\anaddr)} (\pcconf',\valconf[\areg:=\valconf(\anaddr)],\bufconf)$}
	{LM}
	\hspace{-1em}
	\therule{ %
		\acommand\,= \thestore{\anexpr_\anaddr}{\anexpr_\aval} 
	}{$\astate\transitionto{(\athread,\aninstruction,\anaddr)}
		(\pcconf',\valconf,\bufconf[\athread:=(\anaddr,\aval)\cdot\bufconf(\athread)])$}
	{ST}
	\hspace{-1em}
	\doublerule{ %
		$\bufconf(\athread) = \beta\cdot(\anaddr,\aval)$
	}{$\astate\transitionto{(\athread,\writekind,\anaddr)}(\pcconf,\valconf[\anaddr:=\aval],\bufconf[\athread:=\beta])$}
	{ %
		\acommand\,= \themfence{}
		\quad$\bufconf(\athread)=\varepsilon$
	}{$\astate\transitionto{(\athread,\aninstruction,\bottom)}(\pcconf',\valconf,\bufconf)$}
	{MEM}{F}
}
\vskip -1em
\caption{Transition rules for $\eventautomaton{\tso}{\aprogram}$ 
	assuming $\astate=(\pcconf,\valconf,\bufconf)$ 
	with $\pcconf(\athread)=\acontrolstate$ and $\aninstruction=(\acontrolstate, \acommand, \acontrolstate')$ in thread $\athread$.
The program counter is always set to $\pcconf'=\pcconf[\athread:=\acontrolstate']$.
We assume $\anaddr=\valOf{\anexpr_\anaddr}$ to be the address returned by an address expression $\anexpr_\anaddr$ and $\aval=\valOf{\anexpr_{\aval}}$ the value returned by a value expression $\anexpr_{\aval}$.
	We use $\projection{\bufconf(\athread)}{(\set{\anaddr}\times\datadomain)}$ to project the buffer content $\bufconf(\athread)$ to store operations that access address $\anaddr$.}
\label{Figure:TSORules}
\end{figure}

The \seqcon~\cite{Lamport79} semantics is simpler than \tso\ in that stores are not buffered.
Technically, we keep the set of states but change the transitions so that Rule~({ST}) is immediately followed by Rule~(MEM).

We are interested in the \emph{computations} of program $\aprogram$ under $\memmodel\in\{\tso, \seqcon\}$.
They are given by
$\computationsOf{\memmodel}{\aprogram}:=\langFinOf{\finalstates}{\eventautomaton{\memmodel}{\aprogram}}$,
where $\finalstates$ is the set of states with empty buffers.
With this choice of final states, we avoid incomplete computations that have pending stores.
Note that all \seqcon\ states have empty buffers, which means the \seqcon{} computations form a subset of the \tso{} computations:
$\computationsOf{\seqcon}{\aprogram}\subseteq\computationsOf{\tso}{\aprogram}$.
We will use notation $\reachOf{\memmodel}{\aprogram}$ for the set of all states $\astate\in\finaleventstates$
that are reachable by some computation in $\computationsOf{\memmodel}{\aprogram}$.

To give an example, the program from Figure~\ref{Figure:Dekker} admits the \tso\ computation $\tau_\wit$ below
where the store of the first thread is flushed at the end:
\begin{align*}
\tau_\wit\ &=\ \one{\storeevent}\cdot\one{\loadevent}\cdot
	\two{\storeevent}\cdot\two{\flushevent}\cdot\two{\loadevent}\cdot
	\one{\flushevent}.
\end{align*}

Consider an event $\anevent=(\athread,\aninstruction,\anaddr)$.
By $\threadOf{\anevent}:=\athread$ we refer to the thread that produced the event.
Function $\instructionOf{\anevent}:=\aninstruction$ returns the instruction.
For flush events, $\instructionOf{\anevent}$ gives the instruction of the matching store event.
By $\addrOf{\anevent}:=\anaddr$ we denote the address that is accessed (if any).
In the example, $\threadOf{\one{\storeevent}}=\athread_1$,
$\instructionOf{\one{\storeevent}}=\acontrolstateOf{0,1}\transitionto{\thestore{\xaddr}{1}}\acontrolstateOf{1,1}$, and $\addrOf{\one{\storeevent}}=\xaddr$.
\section{Lazy TSO Reachability}\label{Section:LazyReachability}
We introduce the reachability problem and present our main contribution:
an algorithm that checks \tso{} reachability lazily.
The iterative algorithm queries an oracle to identify sequences of instructions that, under the \tso{} semantics,
lead to states not reachable under \seqcon.
In Section~\ref{Subsection:SoundAndComplete}, we show that the algorithm yields a sound and complete semi-decision procedure.

Given a memory model $\memmodel\in\set{\seqcon,\tso}$,
the $\memmodel$ reachability problem expects as input a program $P$ and a set of \emph{goal states} $\goalstates\subseteq\eventstatesOf{\memmodel}$.
We are mostly interested in the control state of each thread. Therefore, goal states $(\pcconf, \valconf, \bufconf)$ typically specify a program counter $\pcconf$ but leave the memory valuation unconstrained.
Formally, the \emph{\memmodel{} reachability problem} asks if some state in $\goalstates$ is reachable in the automaton
$\eventautomaton{\memmodel}{\aprogram}$.
\begin{quote}
{\bf Given:} A parallel program $\aprog$ and goal states $\goalstates$.\\
{\bf Problem:} Decide $	\langFinOf{\finalstates \cap \goalstates}{\eventautomaton{\memmodel}{\aprogram}}\ \neq\ \emptyset$.
\end{quote}
We use notation $\reachOf{\memmodel}{\aprogram}\cap\goalstates$ for the set of reachable final goal states in $\aprogram$.

Instead of solving reachability under \tso{} directly, the algorithm we propose solves \seqcon{} reachability and,
if no goal state is reachable, tries to lazily introduce store buffering on a certain control path of the program.
The algorithm delegates choosing the control path to an \emph{oracle function $\oracle$}.
Given an input program $\bprog$, the oracle returns a sequence of instructions $\instructions^*$ in that program.
Formally, the oracle satisfies the following requirements:
\begin{itemize}
\item If $\oracle(\bprog)=\varepsilon$ then $\reachOf{\seqcon}{\bprog}=\reachOf{\tso}{\bprog}$.
\item Otherwise, $\oracle(\bprog)=\aninstruction_1\aninstruction_2\ldots\aninstruction_n$ with $\commandOf{\aninstruction_1}$ a store, $\commandOf{\aninstruction_n}$ a load, $\commandOf{\aninstruction_i}\neq\themfence$,
and $\destinationstateOf{\aninstruction_i}=\sourcestateOf{\aninstruction_{i+1}}$ for $i\in\intrange{1}{n-1}$.
\end{itemize}

The lazy \tso{} reachability checker is outlined in Algorithm~\ref{Algorithm:Reachability}.
As input, it takes a program $\aprog$ and an oracle $\oracle$.
We assume some control states in each thread to be marked to define a set of goal states.
The algorithm  returns \mytrue{} iff the program can reach a goal state under \tso.
It works as follows.
First, it creates a copy $\bprog$ of the program $\aprog$.
Next, it checks if a goal state is \seqcon{}-reachable in $\bprog$ (Line~\ref{Algorithm:SCReachabilityCheck}).
If that is the case, the algorithm returns \mytrue{}.
Otherwise, it asks the oracle $\oracle$ where in the program to introduce store buffering.
If $\oracleOf{\bprog}\neq\varepsilon$, the algorithm extends $\bprog$ to emulate store buffering on the path $\oracleOf{\bprog}$ under \seqcon{}
(Line~\ref{Algorithm:Extend}).
Then it goes back to the beginning of the loop.
If $\oracleOf{\bprog}=\varepsilon$, by the first property of oracles, $\bprog$ has the same reachable states under \seqcon{} and under \tso{}.
This means the algorithm can safely return \myfalse{} (Line~\ref{Algorithm:NotReachable}).
Note that, since $\bprog$ emulates \tso{} behavior of $\aprog$, the algorithm solves \tso{} reachability for $\aprog$.

\begin{center}
	\begin{algorithm}
	\caption{Lazy \tso{} reachability Checker\label{Algorithm:Reachability}}
	\begin{algorithmic}[1]
	\Statex {\bf Input}: Marked program $\aprog$ and oracle $\oracle$
	\Statex {\bf Output}: \mytrue\ if some goal state is \tso{}-reachable in $\aprog$
	\Statex \hspace{40pt} \myfalse\ if no goal state is \tso{}-reachable in $\aprog$\vspace{0.2cm}
	\State $\bprog:=\aprog$;\label{Algorithm:MakeACopy}
	\State {\bf while} \mytrue{} {\bf do}
		\State\quad {\bf if} $\reachOf{\seqcon}{\aprogram}\cap\goalstates\neq\emptyset$ {\bf then}
		\Comment{check if some goal state is \seqcon{}-reachable}\label{Algorithm:SCReachabilityCheck}
			\State\qquad {\bf return} \mytrue; \label{Algorithm:Reachable}
		\State\quad {\bf else}
			\State\qquad $\reowit:=\oracleOf{\bprog}$;
			\Comment{ask the oracle where to use store buffering}
			\State\qquad {\bf if} $\reowit\neq\varepsilon$ {\bf then}
			\label{Algorithm:RobustnessCheck}
			\State\quad\qquad $\bprog:=\bprog\glue\reowit$; \label{Algorithm:Extend}
		\State\qquad {\bf else } \label{Algorithm:Robust}
			\State\quad\qquad {\bf return} \myfalse; \label{Algorithm:NotReachable}
	\end{algorithmic}
	\end{algorithm}
\end{center}

Let $\reowit:=\oracleOf{\bprog}=\aninstruction_1\aninstruction_2\ldots\aninstruction_n$
and let $\attackthread :=(
  \commands_{\attackthread},
  \controlstates_{\attackthread},
  \instructions_{\attackthread},
  \initialcontrolstateOf{\attackthread})$
be the thread of the instructions in $\reowit$.
The modified program $\bprog\glue\reowit$
replaces $\attackthread$ by a new thread $\attackthread\glue\reowit$.
The new thread emulates under \seqcon{} the \tso{} semantics of $\reowit$.
Formally, the \emph{extension of $\athread$ by $\reowit$} is $\attackthread\glue\reowit:=(
  \commands_{\attackthread}',
  \controlstates_{\attackthread}',
  \instructions_{\attackthread}',
  \initialcontrolstateOf{\attackthread})$.
The thread is obtained from $\attackthread$ by adding sequences of instructions starting from $\extensionstate{0}:=\sourcestateOf{\aninstruction_1}$.
To remember the addresses and values of the buffered stores,
we use auxiliary registers $\adrreg_{1},\ldots,\adrreg_{\bound}$ and $\valreg_{1},\ldots,\valreg_{\bound}$,
where $\bound\leq n-1$ is the total number of store instructions in $\reowit$.
The sets $\commands_{\attackthread}'\supseteq\commands_{\attackthread}$
and $\controlstates_{\attackthread}'\supseteq\controlstates_{\attackthread}$ are extended as necessary.

We define the extension by describing the new transitions that are added to $\instructions_{\attackthread}'$ for each $\aninstruction_i$.
In our construction, we use a variable $\iterator$ to keep track of the number of store instructions already processed.
Initially, $\controlstates_{\attackthread}':=\controlstates_{\attackthread}$ and $\iterator:=0$.
Based on the type of instructions, we distinguish the following cases.

If $\commandOf{\aninstruction_i}=\thestore{\anexpr}{\anexpr'}$, we increment $\iterator$ by $1$
and add instructions that remember the address and the value being written in $\adrreg_\iterator$ and $\valreg_\iterator$.

If $\commandOf{\aninstruction_i}=\theload{\areg}{\anexpr}$, we add instructions to $\instructions_{\attackthread}'$
that perform a load from memory only when a load from the simulated buffer is not possible.
More precisely, if $j\in[1,\iterator]$ is found so that $\adrreg_j=\anexpr$,
register $\areg$ is assigned the value of $\valreg_j$.
Otherwise, $\areg$ receives its value from the address indicated by $\anexpr$.
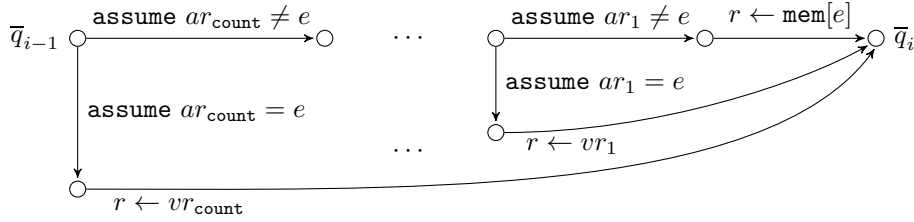
\begin{center}
	\begin{tikzpicture}[->,>=stealth',shorten >=1pt, auto, node distance=2.25cm, transform shape, scale=1]%
		\tikzstyle{every state}=[fill=white,circle,draw=black,inner sep=0pt,text=black,minimum size=6pt]
		\node[state,label=left:$\extensionstate{i-1}$] (qi-1) {};
		\node[state] (fresh1) [node distance=3.25cm,right of=qi-1] {};
		\node[state] (fresh2) [right of=fresh1] {};
		\node (dots1) [node distance=1.125cm,right of=fresh1] {$\cdots$};		
		\node[state] (fresh3) [node distance=2.75cm,right of=fresh2] {};
		\node[state] (fresh4) [node distance=1.25cm,below of=fresh2] {};
		\node[state] (fresh5) [node distance=2cm,below of=qi-1] {};
		\node (dots2) [node distance=1.5cm,below of=dots1] {$\cdots$};
		\node[state,label=right:$\extensionstate{i}$] (qi) [right of=fresh3] {};
		\path	(qi-1)		edge	node {$\theassume{\adrreg_\iterator\neq\anexpr}$} (fresh1)
				(fresh2)	edge 	node {$\theassume{\adrreg_1\neq\anexpr}$} (fresh3)
				(fresh3)	edge	node {$\theload{\areg}{\anexpr}$} (qi)
				(fresh2)	edge 	node {$\theassume{\adrreg_1=\anexpr}$} (fresh4)
				(qi-1)		edge node [] {$\theassume{\adrreg_\iterator=\anexpr}$} (fresh5); %
		\draw (fresh4.east) .. controls +(0:2) and +(210:1) .. (qi.south west) node[below,pos=.15] {$\thelocal{\areg}{\valreg_1}$};
		\draw (fresh5.east) .. controls +(0:4) and +(235:2.5) .. (qi.south) node[below,pos=.1] {$\thelocal{\areg}{\valreg_\iterator}$};
	\end{tikzpicture}
\end{center}
\vskip -.75em

If $\commandOf{\aninstruction_i}$ is an assignment or a conditional,
we add $\atransition{\extensionstate{i-1}}{\extensionstate{i}}{\commandOf{\aninstruction_i}}$ to $\instructions_{\attackthread}'$.
By the definition of an oracle, $\commandOf{\aninstruction_i}$ is never a fence.

The above cases handle all instructions in $\reowit$.
So far, the extension added new instructions to $\instructions_{\attackthread}'$ that lead through the fresh states $\extensionstate{1},\ldots,\extensionstate{n}$.
Out of control state $\extensionstate{n}$, we now recreate the sequence of stores remembered by the auxiliary registers.
Then we return to the control flow of the original thread $\attackthread$.
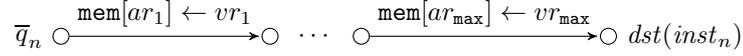
\begin{center}
	\begin{tikzpicture}[->,>=stealth',shorten >=1pt, auto, node distance=3.25cm, transform shape, scale=1]%
		\tikzstyle{every state}=[fill=white,circle,draw=black,inner sep=0pt,text=black,minimum size=6pt]
		\node[state,label=left:$\extensionstate{n}$] (qn) {};
		\node[state] (fresh1) [node distance=2.75cm,right of=qn] {};
		\node (dots) [node distance=.6cm,right of=fresh1] {$\cdots$};
		\node[state] (fresh2) [node distance=.6cm,right of=dots] {};
		\node[state,label=right:$\destinationstateOf{\aninstruction_n}$] (q) [right of=fresh2] {};
		\path	(qn)		edge	node {$\thestore{\adrreg_1}{\valreg_1}$} (fresh1)
				(fresh2)	edge 	node {$\thestore{\adrreg_\bound}{\valreg_\bound}$} (q);
	\end{tikzpicture}
\end{center}
\vskip -.75em

Next, we remove $\aninstruction_1$ from the program.
This prevents the oracle from discovering in the future another instruction sequence that is essentially the same as $\reowit$.
As we will show, this is key to guaranteeing termination of the algorithm for acyclic programs.
However, the removal of $\aninstruction_1$ may reduce the set of \tso{}-reachable states.
To overcome this problem, we insert additional instructions.
Consider an instruction $\aninstruction\in\instructions_{\attackthread}$ with $\sourcestateOf{\aninstruction}=\sourcestateOf{\aninstruction_i}$
for some $i\in\intrange{1}{n}$ and assume that $\aninstruction\neq\aninstruction_i$.
We add instructions that recreate the stores buffered in the auxiliary registers and return to $\destinationstateOf{\aninstruction}$.
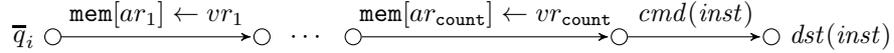
\begin{center}
	\begin{tikzpicture}[->,>=stealth',shorten >=1pt, auto, node distance=3.5cm, transform shape, scale=1]%
		\tikzstyle{every state}=[fill=white,circle,draw=black,inner sep=0pt,text=black,minimum size=6pt]
		\node[state,label=left:$\extensionstate{i}$] (qi) {};
		\node[state] (fresh1) [node distance=2.75cm,right of=qi] {};
		\node (dots) [node distance=.6cm,right of=fresh1] {$\cdots$};
		\node[state] (fresh2) [node distance=.6cm,right of=dots] {};
		\node[state] (fresh3) [right of=fresh2] {};
		\node[state,label=right:$\destinationstateOf{\aninstruction}$] (q) [node distance=2cm,right of=fresh3] {};
		\path	(qi)		edge	node {$\thestore{\adrreg_1}{\valreg_1}$} (fresh1)
				(fresh2)	edge 	node {$\thestore{\adrreg_\iterator}{\valreg_\iterator}$} (fresh3)
				(fresh3)	edge 	node {$\commandOf{\aninstruction}$} (q);
	\end{tikzpicture}
\end{center}
\vskip -.75em

Similarly, for all load instructions $\aninstruction_i$
as well as out of $\extensionstate{1}$
we add instructions that flush and fence the pair $(\adrreg_1,\valreg_1)$,
make visible the remaining buffered stores,
and return to state $\acontrolstate$ in the original control flow.
Below, $\acontrolstate:=\sourcestateOf{\aninstruction_i}$ if $\aninstruction_i$ is a load and $\acontrolstate:=\destinationstateOf{\aninstruction_1}$, otherwise.
Intuitively, this captures behaviors that delay $\aninstruction_1$ past loads earlier than $\aninstruction_n$, and that do not delay $\aninstruction_1$ past the first load in $\reowit$.
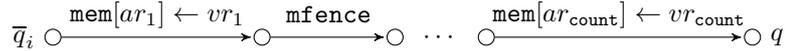
\begin{center}
	\begin{tikzpicture}[->,>=stealth',shorten >=1pt, auto, node distance=3.5cm, transform shape, scale=1]%
		\tikzstyle{every state}=[fill=white,circle,draw=black,inner sep=0pt,text=black,minimum size=6pt]
		\node[state,label=left:$\extensionstate{i}$] (qi) {};
		\node[state] (fresh1) [node distance=2.75cm,right of=qi] {};
		\node[state] (fresh2) [node distance=1.75cm,right of=fresh1] {};
		\node (dots) [node distance=.6cm,right of=fresh2] {$\cdots$};
		\node[state] (fresh3) [node distance=.6cm,right of=dots] {};
		\node[state,label=right:$\acontrolstate$] (q) [node distance=3.5cm,right of=fresh3] {};
		\path	(qi)		edge	node {$\thestore{\adrreg_1}{\valreg_1}$} (fresh1)
				(fresh1)	edge 	node [yshift=2pt] {$\themfence$} (fresh2)
				(fresh3)	edge 	node {$\thestore{\adrreg_\iterator}{\valreg_\iterator}$} (q);
	\end{tikzpicture}
\end{center}

Figure~\ref{Figure:Instrumented} shows the extension of the program in Figure~\ref{Figure:Dekker} by the instruction sequence
$\instructionOf{\one{\storeevent}}\cdot\instructionOf{\one{\loadevent}}:=
\acontrolstateOf{0,1}\transitionto{\thestore{\xaddr}{1}}\acontrolstateOf{1,1}\transitionto{\theload{\areg_1}{\yaddr}}\acontrolstateOf{1,2}$.

\begin{figure}[!t]
	\centering\small\vskip -1em
	\begin{tikzpicture}[->,>=stealth',shorten >=1pt, auto, node distance=1cm, transform shape, scale=1,baseline=.65cm]%
		\tikzstyle{every state}=[fill=white,circle,draw=black,inner sep=0pt,text=black,minimum size=6pt]
		\begin{scope} %
			\node[state,initial,initial above,initial text=$\athread_1$,label=left:$\acontrolstate_{0,1}$] (q10) {};
			\node[state,label=left:$\acontrolstate_{1,1}$] (q11) [below of=q10] {};
			\node[state,label=left:$\acontrolstate_{2,1}$] (q12) [below of=q11] {};
			\node[accepting,state,label=left:$\finalcontrolstateOf{1}$] (q13) [below of=q12] {};
			\path	(q11)		edge 	node [swap] {$\theload{r_1}{\yaddr}$} (q12)
					(q12)		edge	node [swap] {$\theassume{r_1\hspace{-.2em}=\hspace{-.1em}0}$} (q13);
					
			\node[state] (qstore) [node distance=2cm,right of=q10] {};
			\node[state,label=above:$\extensionstate{1}$] (q1) [node distance=2cm,right of=qstore] {};
			\node[state] (qbuf) [below of=q1] {};
			\node[state] (qmem) [node distance=2.75cm,right of=q1] {};
			\node[state,label=below:$\extensionstate{2}$] (q2) [below of=qbuf] {};

			\path (q10)		edge node {$\thelocal{\adrreg_1}{\xaddr}$} (qstore)
					(qstore)	edge node {$\thelocal{\valreg_1}{1}$} (q1)
					(q1)		edge node {$\theassume{\adrreg_1\neq\yaddr}$} (qmem)
					(q1)		edge node {$\theassume{\adrreg_1=\yaddr}$} (qbuf)
					(qbuf)	edge node {$\thelocal{\areg_1}{\valreg_1}$} (q2)
					(q2)		edge node {$\thestore{\adrreg_1}{\valreg_1}$} (q12);
			\draw (qmem.south) .. controls +(270:.5) and +(0:2.5) .. (q2.east) node[pos=.75,xshift=-10pt,yshift=0pt] {$\theload{\areg_1}{\yaddr}$};

			\node[state] (qfence) [node distance=1.5cm,right of=q11] {};
			\path (q1)		edge node [below,sloped] {$\thestore{\adrreg_1}{\valreg_1}$} (qfence)
					(qfence)	edge node {$\themfence$} (q11);
		\end{scope}
		\begin{scope} %
			\node[state,initial,initial above,initial text=$\athread_2$,label=right:$\acontrolstate_{0,2}$] (q20) [node distance=2.5cm,right of=qmem] {};
			\node[state,label=right:$\acontrolstate_{1,2}$] (q21) [below of=q20] {};
			\node[state,label=right:$\acontrolstate_{2,2}$] (q22) [below of=q21] {};
			\node[accepting,state,label=right:$\finalcontrolstateOf{2}$] (q23) [below of=q22] {};
			\path	(q20)		edge	node [swap] {$\thestore{\yaddr}{1}$} (q21)
					(q21)		edge 	node [swap] {$\theload{r_2}{\xaddr}$} (q22)
					(q22)		edge	node [swap] {$\theassume{r_2\hspace{-.2em}=\hspace{-.1em}0}$} (q23);
		\end{scope}	
	\end{tikzpicture}
\caption{Extension by $\instructionOf{\one{\storeevent}}\cdot\instructionOf{\one{\loadevent}}$ of the program in Figure~\ref{Figure:Dekker}. 
Goal state $(\pcconf,\valconf,\bufconf)$ with $\valconf(\xaddr)=\valconf(\yaddr)=1$ and $\valconf(\areg_1)=\valconf(\areg_2)=0$ is now \seqcon{}-reachable.}
\label{Figure:Instrumented}
\end{figure}

\subsection{Soundness and Completeness}\label{Subsection:SoundAndComplete}
We show that Algorithm~\ref{Algorithm:Reachability} is a decision procedure for acyclic programs.
From here until (inclusively) Theorem~\ref{Theorem:Correctness} we assume that all programs are acyclic,
i.e.,\ their instructions and control states form directed acyclic graphs.
Theorem~\ref{Theorem:PartialCorrectness} then explains how Algorithm~\ref{Algorithm:Reachability} yields a semi-decision procedure for all programs.

We first prove the extension sound and complete (Lemma~\ref{Lemma:Extension}):
extending $\bprog$ by sequence $\reowit:=\oracleOf{\bprog}$ does neither add nor remove \tso{}-reachable states.
Afterwards, Lemma~\ref{Lemma:NewWitnessSupport} shows that
if Algorithm~\ref{Algorithm:Reachability} extends $\bprog$ by $\reowit$ (Line~\ref{Algorithm:Extend}) then,
in subsequent iterations of the algorithm, no new sequence returned by the oracle is the same as $\reowit$ (projected back to $\aprog$).
Next, by the first condition of an oracle and using Lemma~\ref{Lemma:NewWitnessSupport}, we establish that
Algorithm~\ref{Algorithm:Reachability} is a decision procedure for acyclic programs (Theorem~\ref{Theorem:Correctness}).
Finally, we show that Algorithm~\ref{Algorithm:Reachability} can be turned into a semi-decision procedure for all programs using a bounded model checking approach
(Theorem~\ref{Theorem:PartialCorrectness}).

\begin{alemma}\label{Lemma:Extension}
Let $\addrdomain\cup\regdomain$ be the addresses and registers of program $\bprog$ and let $\reowit:=\oracleOf{\bprog}$.
Then we have
$(\pcconf,\valconf,\bufconf)\in\reachOf{\tso}{\bprog}$ if and only if
$(\pcconf,\valconf',\bufconf)\in\reachOf{\tso}{\bprog\glue\reowit}$ with
$\valconf(\anaddr)=\valconf'(\anaddr)$ for all $\anaddr\in\addrdomain\cup\regdomain$.
\end{alemma}
Let $\athread$ be the thread that differs in $\bprog$ and $\bprog\glue\reowit$.
To prove Lemma~\ref{Lemma:Extension}, one can show that
for any prefix $\acomp'$ of $\acomp\in\computationsOf{\tso}{\bprog}$
there is a prefix $\bcomp'$ of $\bcomp\in\computationsOf{\tso}{\bprog\glue\reowit}$, and vice versa, that maintain the following invariants.
\begin{asparaitem}
\item[{\bf Inv-0}] $\initialeventstate\transitionto{\acomp'}(\pcconf,\valconf,\bufconf)$
and $\initialeventstate\transitionto{\bcomp'}(\pcconf',\valconf',\bufconf')$.\vspace{0.1cm}
\item[{\bf Inv-1}] If $\pcconf$ and $\pcconf'$ differ, they only differ for thread $\attackthread$.
	If $\pcconf(\attackthread)\neq\pcconf'(\attackthread)$, then
	$\pcconf(\attackthread)=\destinationstateOf{\aninstruction_i}$ and $\pcconf'(\attackthread)=\extensionstate{i}$ for some $i\in\intrange{1}{n-1}$.\vspace{0.1cm}
\item[{\bf Inv-2}] $\valconf'(\anaddr)=\valconf(\anaddr)$ for all $\anaddr\in\addrdomain\cup\regdomain$.\vspace{0.1cm}
\item[{\bf Inv-3}] $\bufconf$ and $\bufconf'$ differ at most for $\attackthread$.
	If $\bufconf(\attackthread)\neq\bufconf'(\attackthread)$, then
	$\pcconf'(\attackthread)=\extensionstate{i}$ for some $i\in\intrange{1}{n-1}$ and
	$\bufconf(\attackthread)=({\valOf{\adrreg_\iterator}},{\valOf{\valreg_\iterator}})
		\cdots({\valOf{\adrreg_1}},{\valOf{\valreg_1}})\cdot\bufconf'(\attackthread)$
	where $\iterator$ stores are seen along $\reowit$ from $\sourcestateOf{\aninstruction_1}$ to $\destinationstateOf{\aninstruction_i}$.\vspace{0.2cm}
\end{asparaitem}

\lemmaatend
Prior to proving Lemma~\ref{Lemma:Extension} we do a bit of preparation. 
We rely on computations that delay flush events locally the least.
Lemma~\ref{Lemma:NoUselessDelays} explains what this means.
\begin{alemma}\label{Lemma:NoUselessDelays}
Let $\acomp\in\computationsOf{\tso}{\bprog}$ and $\attackthread\in\threaddomain$. 
There exists $\ddot{\acomp}\in\computationsOf{\tso}{\bprog}$ such that 
$\happensbeforeOf{\acomp}=\,\happensbeforeOf{\ddot{\acomp}}$ and, 
for all events $\anevent_\storeevent\equivalenceorder\anevent_\flushevent$ within thread $\attackthread$, 
if $\projection{\ddot{\acomp}}{\attackthread}:=
  \acomp_\text{prefix}\cdot\anevent_\storeevent\cdot\acomp'\cdot\anevent_\flushevent\cdot\acomp_\text{suffix}$ then either 
\begin{asparaitem}
\item[(1)] $\acomp':=\bcomp\cdot\anevent_\loadevent\cdot\bcomp'$ and all events $\anevent\in\bcomp'$ are flushes,
\item[{\hspace{-14pt}or (2)}] all events $\anevent\in\acomp'$ are local assignments or conditionals
\end{asparaitem}
\end{alemma}
\begin{proof}
Intuitively, the theorem states that flush events of thread $\attackthread$ delayed past same-thread local events, 
may be delayed less without changing the happens-before relation of the computation.
Local events are assignments, conditionals, and store events in the same thread.

Let $\acomp:=\acomp_1\cdot\anevent_\storeevent\cdot\acomp_2\cdot\anevent\cdot\acomp_3\cdot\anevent_\flushevent\cdot\acomp_4$
such that $\anevent_\storeevent\equivalenceorder\anevent_\flushevent$ are events of thread $\attackthread$,
$\anevent$ is a local event in $\attackthread$
and $\threadOf{\anevent'}\neq\attackthread$ for all events $\anevent'\in\acomp_3$.

We denote by 
$\acomp_0:=\acomp_1\cdot\anevent_\storeevent\cdot\acomp_2\cdot\acomp_3\cdot\anevent_\flushevent\cdot\anevent\cdot\acomp_4$
the \tso{} computation that first performs the flush $\anevent_\flushevent$ and then the event $\anevent$.
Notice that since $\acomp_3$ contains no events $\anevent'$ with $\threadOf{\anevent'}=\attackthread$, 
feasibility of computation $\acomp_0$ is ensured and $\happensbeforeOf{\acomp}=\,\happensbeforeOf{\acomp_0}$ holds.

Starting with the last flush event in $\acomp$, we use the above reordering of events $\anevent$ to locally delay flush events less.
In the end we obtain computation $\ddot{\acomp}$ in which no flush event of thread $\attackthread$ can be locally delayed less.
\qed
\end{proof}

Furthermore, in order to reference instructions of $\bprog\glue\reowit$ that the extension adds 
we give an alternative description for some of the transition sequences in the main text.
Recall that variable $\iterator$ keeps track of the number of store instructions processed along $\reowit$.

If $\commandOf{\aninstruction_i}=\thestore{\anexpr}{\anexpr'}$, we said $\iterator$ is incremented 
and instructions that remember the value and address written in $\adrreg_\iterator$ and $\valreg_\iterator$ are added.
\begin{align}
	\begin{tikzpicture}[->,>=stealth',shorten >=1pt, auto, node distance=2.25cm, transform shape, scale=1,baseline=-3pt]%
		\tikzstyle{every state}=[fill=white,circle,draw=black,inner sep=0pt,text=black,minimum size=6pt]
		\node[state,label=left:$\extensionstate{i-1}$] (qi-1) {};
		\node[state] (fresh) [right of=qi-1] {};
		\node[state,label=right:$\extensionstate{i}$] (qi) [right of=fresh] {};
		\path	(qi-1)		edge	node {$\thelocal{\adrreg_\iterator}{\anexpr}$} (fresh)
				(fresh)	edge 	node {$\thelocal{\valreg_\iterator}{\anexpr'}$} (qi);
	\end{tikzpicture}
	\label{Store:AddressValue}
\end{align}

If $\commandOf{\aninstruction_i}=\theload{\areg}{\anexpr}$ we said instructions are added 
that load from memory only when a load from the simulated buffer is not possible.
More precisely, if some $j\in[1,\iterator]$ such that $\adrreg_j=\anexpr$ is found, 
$\areg$ is assigned the value of $\valreg_j$.
Otherwise, the register $\areg$ receives its value from the address $\valOf{\anexpr}$.
\begin{center}
	\begin{tikzpicture}[->,>=stealth',shorten >=1pt, auto, node distance=2.25cm, transform shape, scale=1]%
		\tikzstyle{every state}=[fill=white,circle,draw=black,inner sep=0pt,text=black,minimum size=6pt]
		\node[state,label=left:$\extensionstate{i-1}$] (qi-1) {};
		\node[state] (fresh1) [node distance=3.25cm,right of=qi-1] {};
		\node[state] (fresh2) [right of=fresh1] {};
		\node (dots1) [node distance=1.125cm,right of=fresh1] {$\cdots$};		
		\node[state] (fresh3) [node distance=2.75cm,right of=fresh2] {};
		\node[state] (fresh4) [node distance=1.5cm,below of=fresh2] {};
		\node[state] (fresh5) [node distance=2.5cm,below of=qi-1] {};
		\node (dots2) [node distance=1.5cm,below of=dots1] {$\cdots$};
		\node[state,label=right:$\extensionstate{i}$] (qi) [right of=fresh3] {};
		\path	(qi-1)		edge	node {$\theassume{\adrreg_\iterator\neq\anexpr}$} (fresh1)
				(fresh2)	edge 	node {$\theassume{\adrreg_1\neq\anexpr}$} (fresh3)
				(fresh3)	edge	node {$\theload{\areg}{\anexpr}$} (qi)
				(fresh2)	edge 	node {$\theassume{\adrreg_1=\anexpr}$} (fresh4)
				(qi-1)		edge node [pos=.275] {$\theassume{\adrreg_\iterator=\anexpr}$} (fresh5);
		\draw (fresh4.east) .. controls +(0:2) and +(210:1) .. (qi.south west) node[below,pos=.15] {$\thelocal{\areg}{\valreg_1}$};		%
		\draw (fresh5.east) .. controls +(0:4) and +(235:3) .. (qi.south) node[below,pos=.1] {$\thelocal{\areg}{\valreg_\iterator}$};		%
	\end{tikzpicture}
\end{center}
Alternatively, assuming $\extensionstate{\text{check},i,\iterator}:=\extensionstate{i-1}$, this can be stated as adding
\begin{align}
	&\hspace{1.45em}
		\set{\atransition{\extensionstate{\text{check},i,\iterator}}{\extensionstate{\text{buf},i,\iterator}}{\theassume{\adrreg_\iterator=\anexpr}}}
		\label{Load:CheckPositiveFirst}\\
	&\disunion~
		\set{\atransition{\extensionstate{\text{check},i,\iterator}}{\extensionstate{\text{check},i,\iterator-1}}{\theassume{\adrreg_\iterator\neq\anexpr}}}
		\label{Load:CheckNegativeFirst}\\	
	&\disunion~
		\set{\atransition{\extensionstate{\text{buf},i,\iterator}}{\extensionstate{i}}{\thelocal{\areg}{\valreg_\iterator}}}\label{Load:BufferFirst}\\
	&\hspace{4.5pt}\vdots~\nonumber\\
	&\disunion~
		\set{\atransition{\extensionstate{\text{check},i,1}}{\extensionstate{\text{buf},i,1}}{\theassume{\adrreg_1=\anexpr}}}\label{Load:CheckPositiveLast}\\
	&\disunion~
		\set{\atransition{\extensionstate{\text{check},i,1}}{\extensionstate{\text{mem},i}}{\theassume{\adrreg_1\neq\anexpr}}}\label{Load:CheckNegativeLast}\\	
	&\disunion~
		\set{\atransition{\extensionstate{\text{buf},i,1}}{\extensionstate{i}}{\thelocal{\areg}{\valreg_1}}}\label{Load:BufferLast}\\
		\nonumber\\
	&\disunion~
		\set{\atransition{\extensionstate{\text{mem},i}}{\extensionstate{i}}{\theload{\areg}{\anexpr}}}\label{Load:Memory}
	\end{align}

We said that out of control state $\extensionstate{n}$ we create a sequence of stores to flush the contents of the auxiliary registers 
and return to the code of the original thread.
\begin{center}
	\begin{tikzpicture}[->,>=stealth',shorten >=1pt, auto, node distance=3.25cm, transform shape, scale=1]%
		\tikzstyle{every state}=[fill=white,circle,draw=black,inner sep=0pt,text=black,minimum size=6pt]
		\node[state,label=left:$\extensionstate{n}$] (qn) {};
		\node[state] (fresh1) [node distance=2.75cm,right of=qn] {};
		\node (dots) [node distance=.6cm,right of=fresh1] {$\cdots$};
		\node[state] (fresh2) [node distance=.6cm,right of=dots] {}; 
		\node[state,label=right:$\destinationstateOf{\aninstruction_n}$] (q) [right of=fresh2] {};
		\path	(qn)		edge	node {$\thestore{\adrreg_1}{\valreg_1}$} (fresh1)
				(fresh2)	edge 	node {$\thestore{\adrreg_\bound}{\valreg_\bound}$} (q);
	\end{tikzpicture}
\end{center}
Alternatively, we could have stated it as adding
\begin{align}
	&\hspace{1.45em}
		\set{\atransition{\extensionstate{n}}{\extensionstate{\text{flush},1}}{\thestore{\adrreg_1}{\valreg_1}}}\label{Flush:First}\\	
	&\hspace{5pt}\vdots&&~\nonumber\\
	&\disunion~
		\set{\atransition{\extensionstate{\text{flush},\bound-1}}{\destinationstateOf{\aninstruction_n}}{\thestore{\adrreg_\bound}{\valreg_\bound}}}	
		\label{Flush:Last}
\end{align}

Furthermore, for all instructions $\aninstruction\in\instructions_{\attackthread}$ with $\sourcestateOf{\aninstruction}=\sourcestateOf{\aninstruction_i}$ 
for some $i\in\intrange{1}{n}$ and for which $\aninstruction\neq\aninstruction_i$ 
we added instructions that flush the stores buffered in the auxiliary registers and return to $\destinationstateOf{\aninstruction}$.
\begin{center}
	\begin{tikzpicture}[->,>=stealth',shorten >=1pt, auto, node distance=3.5cm, transform shape, scale=1]%
		\tikzstyle{every state}=[fill=white,circle,draw=black,inner sep=0pt,text=black,minimum size=6pt]
		\node[state,label=left:$\extensionstate{i}$] (qi) {};
		\node[state] (fresh1) [node distance=2.75cm,right of=qi] {};
		\node (dots) [node distance=.6cm,right of=fresh1] {$\cdots$};
		\node[state] (fresh2) [node distance=.6cm,right of=dots] {}; 
		\node[state] (fresh3) [right of=fresh2] {};
		\node[state,label=right:$\destinationstateOf{\aninstruction}$] (q) [node distance=2cm,right of=fresh3] {};
		\path	(qi)		edge	node {$\thestore{\adrreg_1}{\valreg_1}$} (fresh1)
				(fresh2)	edge 	node {$\thestore{\adrreg_\iterator}{\valreg_\iterator}$} (fresh3)
				(fresh3)	edge 	node {$\commandOf{\aninstruction}$} (q);
	\end{tikzpicture}
\end{center}
Alternatively, we could have stated it as adding
\begin{align}
	&\hspace{1.45em}
		\set{\atransition{\extensionstate{i}}{\extensionstate{\text{next},i,1}}{\thestore{\adrreg_1}{\valreg_1}}}\label{Flush:NextFirst}\\	
	&\hspace{5pt}\vdots&&~\nonumber\\
	&\disunion~
		\set{\atransition{\extensionstate{\text{next},i,\iterator-1}}{\extensionstate{\text{next},i,\iterator}}{\thestore{\adrreg_\iterator}{\valreg_\iterator}}}
		\label{Flush:NextLast}\\
	&\disunion~
		\set{\atransition{\extensionstate{\text{next},i,\iterator}}{\destinationstateOf{\aninstruction}}{\commandOf{\aninstruction}}}
		\label{Another:Inst}
\end{align}

Finally, for all load instructions $\aninstruction_i$, where $i<n$, as well as out of $\extensionstate{1}$ 
we added instructions that flush and fence the pair $(\adrreg_1,\valreg_1)$, 
make the remaining buffered stores in the auxiliary registers visible, 
and return to $\acontrolstate$. 
Here $\acontrolstate:=\sourcestateOf{\aninstruction_i}$ in the load case 
and $\acontrolstate:=\destinationstateOf{\aninstruction_1}$ otherwise.
\begin{center}
	\begin{tikzpicture}[->,>=stealth',shorten >=1pt, auto, node distance=3.5cm, transform shape, scale=1]%
		\tikzstyle{every state}=[fill=white,circle,draw=black,inner sep=0pt,text=black,minimum size=6pt]
		\node[state,label=left:$\extensionstate{i}$] (qi) {};
		\node[state] (fresh1) [node distance=2.75cm,right of=qi] {};
		\node[state] (fresh2) [node distance=1.75cm,right of=fresh1] {}; 
		\node (dots) [node distance=.6cm,right of=fresh2] {$\cdots$};
		\node[state] (fresh3) [node distance=.6cm,right of=dots] {};
		\node[state,label=right:$\acontrolstate$] (q) [node distance=3.5cm,right of=fresh3] {};
		\path	(qi)		edge	node {$\thestore{\adrreg_1}{\valreg_1}$} (fresh1)
				(fresh1)	edge 	node [yshift=2pt] {$\themfence$} (fresh2)
				(fresh3)	edge 	node {$\thestore{\adrreg_\iterator}{\valreg_\iterator}$} (q);
	\end{tikzpicture}
\end{center}
Alternatively, we could have stated it as adding
\begin{align}
	&\hspace{1.45em}
		\set{\atransition{\extensionstate{i}}{\extensionstate{\text{fence},i}}{\thestore{\adrreg_1}{\valreg_1}}}\label{Intermediary:FlushFirst}\\
	&\disunion~	
		\set{\atransition{\extensionstate{\text{fence},i}}{\extensionstate{\text{orig},i,2}}{\themfence}}\label{Intermediary:Fence}\\
	&\disunion~
		\set{\atransition{\extensionstate{\text{orig},i,2}}{\extensionstate{\text{orig},i,3}}{\thestore{\adrreg_2}{\valreg_2}}}\label{Intermediary:FlushSecond}\\
		&\hspace{5pt}\vdots&&~\nonumber\\
	&\disunion~
		\set{\atransition{\extensionstate{\text{orig},i,\iterator}}{\acontrolstate}{\thestore{\adrreg_\iterator}{\valreg_\iterator}}}
		\label{Intermediary:FlushLast}
\end{align}
We can now turn to the actual proof of Lemma~\ref{Lemma:Extension}.
\endlemmaatend

\proofatend
Assume $\attackthread$ is the thread of $\reowit:=\aninstruction_1\cdot\ldots\cdot\aninstruction_n$, 
$\eventautomaton{\tso}{\bprog\glue\reowit}:=(\eventsextended,\statesextended,\eventtransitions{\tso},\initialextended,\finalextended)$,
$\instructions$ and $\controlstates$ are the instructions and states of $\bprog$,
$\addrdomain$ and $\regdomain$ are registers and addresses used by $\bprog$,
and $\instructionsextended$ are the instructions $\instructions_\athread'$ of $\bprog\glue\reowit$ as described in Section~\ref{Section:LazyReachability}.

A direct result of Lemmas~\ref{Lemma:NoUselessDelays}~and~\ref{Lemma:SameTraceReach} is that
\tso{} computations of $\bprog$ that delay flushes of $\attackthread$ locally the least reach all the states in the set $\reachOf{\tso}{\bprog}$.
Assume $\acomp\in\computationsOf{\tso}{\bprog}$ is a computation where flushes of $\attackthread$ are delayed locally the least
as Lemma~\ref{Lemma:NoUselessDelays} describes and let $\astate_0,\ldots,\astate_m\in\eventstatesOf{\tso}$ for some $m\in\nat$ 
be all the states along the transition sequence $\initialeventstate\transitionto{\acomp}\astate$, 
i.e., $\astate_0:=\initialeventstate$ and $\astate_m:=\astate$.
Also, for all $k\in[0,m]$, let $\acomp_k$ denote prefixes of $\acomp$ with $\astate_0\transitionto{\acomp_k}\astate_k$.

We prove by induction over state indexes $k\in[0,m]$ that there exist prefixes $\bcomp_k$ of $\bcomp\in\computationsOf{\tso}{\bprog\glue\reowit}$ 
and states $\astate'_0,\ldots,\astate'_m\in\statesextended$ along 
$\initialextended\transitionto{\bcomp}\astate'\in\eventtransitions{\tso}^*$ with $\astate'_0:=\initialextended$ and $\astate'_m:=\astate'$ 
such that the following invariants hold:
\begin{asparaitem}
\item[{\bf Inv-0}] $\initialeventstate\transitionto{\acomp'}(\pcconf,\valconf,\bufconf)$ 
and $\initialextended\transitionto{\bcomp'}(\pcconf',\valconf',\bufconf')$.
\item[{\bf Inv-1}] If $\pcconf$ and $\pcconf'$ differ then they only differ for thread $\attackthread$. 
	Moreover, if $\pcconf(\attackthread)\neq\pcconf'(\attackthread)$ then 
	$\pcconf(\attackthread)=\destinationstateOf{\aninstruction_i}$ and $\pcconf'(\attackthread)=\extensionstate{i}$ for some $i\in\intrange{1}{n-1}$.
\item[{\bf Inv-2}] $\valconf'(\anaddr)=\valconf(\anaddr)$ for all $\anaddr\in\addrdomain\cup\regdomain$.
\item[{\bf Inv-3}] $\bufconf$ and $\bufconf'$ differ at most for $\attackthread$.
	Furthermore, if $\bufconf(\attackthread)\neq\bufconf'(\attackthread)$ then
	$\pcconf'(\attackthread)=\extensionstate{i}$ for some $i\in\intrange{1}{n-1}$ and
	$\bufconf(\attackthread)=({\valOf{\adrreg_\iterator}},{\valOf{\valreg_\iterator}})
		\cdot\ldots\cdot({\valOf{\adrreg_1}},{\valOf{\valreg_1}})\cdot\bufconf'(\attackthread)$
	where $\iterator$ stores are seen along $\reowit$ from $\sourcestateOf{\aninstruction_1}$ to $\destinationstateOf{\aninstruction_i}$.
\end{asparaitem}

For the induction {\bf base case} $k=0$, 
$\acomp_0=\epsilon$, $\astate_0=\initialeventstate$, $\pcconf=\initpcconf$, $\valconf=\initvalconf$, and $\bufconf=\initbufconf$.
Then, for $\bcomp_0:=\epsilon$ and $\astate_0'=\initialextended$, invariants {\bf Inv-0...3} hold.

For the {\bf induction step case}, assume that invariants {\bf Inv-0...3} hold for $k<m$ and 
that $\astate_k\transitionto{\anevent}\astate_{k+1}:=(\pcconf{}_+,\valconf{}_+,\bufconf{}_+)$
for some $\anevent\in\events$. 
We use a case distinction over possible events $\anevent$ to define $\bcomp_{k+1}$ such that 
$\astate'_0\transitionto{\bcomp_{k+1}}\astate'_{k+1}:=(\pcconf{}_+',\valconf{}_+',\bufconf{}_+')$
and invariants {\bf Inv-0...3} hold for $k+1$.

If \underline{$\threadOf{\anevent}:=\athread'\neq\attackthread$} it means $\instructionOf{\anevent}\in\instructionsextended$
is enabled in $\pcconf'(\athread')$, 
so there exist $\anevent'\in\eventsextended$ and $\astate'_{k+1}\in\statesextended$ such that 
$\instructionOf{\anevent'}:=\instructionOf{\anevent}$ and 
$\atransition{\astate'_k}{\astate'_{k+1}}{\anevent'}\in\eventtransitions{\tso}$ in $\eventautomaton{\tso}{\bprog\glue\reowit}$.
We define $\bcomp_{k+1}:=\bcomp_k\cdot\anevent'$ and find that, 
by the $\eventtransitions{\tso}$ semantics (Figure~\ref{Figure:TSORules}) and under the assumption that invariants {\bf Inv-0...3} hold for $k$, 
invariants {\bf Inv-0...3} also hold for $k+1$.

If \underline{$\threadOf{\anevent}=\attackthread$} we make the following case distinction over $\anevent$ and $\pcconf'(\attackthread)$.
\begin{asparaitem}
	\item[\fbox{1}\,``$\anevent$ is a flush event.''] %
		This first case deals with the possibility that a store operation is flushed. 
		Depending on whether $\bufconf'(\attackthread)\neq\epsilon$, 
		we either flush the oldest address-value pair of $\bufconf'(\attackthread)$ or the first address-value auxiliary registers pair. 
		By Lemma~\ref{Lemma:NoUselessDelays}, the later case can only happen when 
		$\pcconf'(\attackthread)=\extensionstate{i}$ for some $i\in\intrange{2}{n-1}$ and $\aninstruction_i$ performs a load or $i=1$. 
		 
		\quad If $\bufconf'(\attackthread)\neq\epsilon$ we flush the oldest write access buffered.
		Namely, let $\anevent_\flushevent\in\eventsextended$ and $\astate'_{k+1}\in\statesextended$ such that, according to rule ({WM}), 
		$\atransition{\astate'_k}{\astate'_{k+1}}{\anevent_\flushevent}\in\eventtransitions{\tso}$.
		We define $\bcomp_{k+1}:=\bcomp_k\cdot\anevent_\flushevent$ and invariants {\bf Inv-0...3} hold for $k+1$ since
			\begin{asparaitem}
				\item[\quad(0)] {\bf Inv-0,3} hold for $k$ so $\astate_0\transitionto{\acomp_{k+1}}\astate_{k+1}$
					and $\astate'_0\transitionto{\bcomp_{k+1}}\astate'_{k+1}$, implying {\bf Inv-0} holds for $k+1$.
				\item[\quad(1)] {\bf Inv-1} holds for $k$, $\pcconf{}_+(\attackthread)=\pcconf(\attackthread)$,
					and $\pcconf{}_+'(\attackthread)=\pcconf'(\attackthread)$, so {\bf Inv-1} holds for $k+1$.
				\item[\quad(2)] {\bf Inv-2,3} hold for $k$, so events $\anevent$ and $\anevent_\flushevent$ update the same address by a same value 
					and {\bf Inv-2} holds for $k+1$.
				\item[\quad(3)] {\bf Inv-3} holds for $k$ and events $\anevent$ and $\anevent_\flushevent$ remove one address-value pair from both
					$\bufconf(\attackthread)$ and $\bufconf'(\attackthread)$, so {\bf Inv-3} holds for $k+1$.
			\end{asparaitem}

		\quad Otherwise, 	$\bufconf'(\attackthread)=\epsilon$ and $\iterator$ stores are encountered from $\sourcestateOf{\aninstruction_1}$ 
			to $\pcconf'(\attackthread)=\extensionstate{i}$ for some $i\in\intrange{1}{n-1}$.
			Then $\bufconf(\attackthread)=({\valOf{\adrreg_\iterator}},{\valOf{\valreg_\iterator}})
				\cdot\ldots\cdot({\valOf{\adrreg_1}},{\valOf{\valreg_1}})$ and, by Lemma~\ref{Lemma:NoUselessDelays}, 
			we know $\aninstruction_i$ is either the first store $\aninstruction_1$ of $\reowit$ or a load. 
			Either way, let $\anevent_{1},\ldots,\anevent_{\iterator},\anevent_\flushevent,\anevent_\fenceevent\in\eventsextended$
			match equations (\ref{Intermediary:FlushFirst}--\ref{Intermediary:FlushLast}) in the extension
			and $\astate'_{k+1}\in\statesextended$ such that events $\anevent_{j}$ are, for all $j\in\intrange{1}{\iterator}$, 
			the buffering events for the stores (\ref{Intermediary:FlushFirst},\ref{Intermediary:FlushSecond}--\ref{Intermediary:FlushLast}),
			$\anevent_\flushevent$ is the flush event for the store (\ref{Intermediary:FlushFirst}),
			$\anevent_\fenceevent$ is the event for the fence (\ref{Intermediary:Fence}),
			and ${\astate'_k}\transitionto{\anevent_{1}\cdot\anevent_\flushevent\cdot\anevent_\fenceevent\cdot
				\anevent_{2}\cdot\ldots\cdot\anevent_{\iterator}}{\astate'_{k+1}}\in\eventtransitions{\tso}^*$
			according to rules ({ST},{MEM},{F}) in Figure~\ref{Figure:TSORules}. 
			We then define $\bcomp_{k+1}:=\bcomp_k\cdot\anevent_{1}\cdot\anevent_\flushevent\cdot\anevent_\fenceevent\cdot
				\anevent_{2}\cdot\ldots\cdot\anevent_{\iterator}\cdot\anevent$ and find that invariants {\bf Inv-0...3} hold for $k+1$ since
				\begin{asparaitem}
					\item[\quad(0)] {\bf Inv-0,3} hold for $k$ so $\astate_0\transitionto{\acomp_{k+1}}\astate_{k+1}$
					and $\astate'_0\transitionto{\bcomp_{k+1}}\astate'_{k+1}$, i.e. {\bf Inv-0} holds for $k+1$.
					\item[\quad(1)] {\bf Inv-1} holds for $k$ and $\pcconf{}_+(\attackthread)=\acontrolstate=\pcconf{}_+'(\attackthread)$,
						where $\acontrolstate:=\sourcestateOf{\aninstruction_i}$ if $\aninstruction_i$ is a load 
						and $\acontrolstate:=\destinationstateOf{\aninstruction_1}$ otherwise,  
						so {\bf Inv-1} holds for $k+1$.
					\item[\quad(2)] {\bf Inv-2,3} hold for $k$, events $\anevent$ and $\anevent_\flushevent$ update the same address by the same value and,
						since the other events do not update any address, {\bf Inv-2} holds for $k+1$.
					\item[\quad(3)] {\bf Inv-3} holds for $k$, and events $\anevent_{2},\ldots,\anevent_{\iterator}$ place the corresponding
						address-value pairs that match $\bufconf{}_+(\attackthread)$ into $\bufconf{}_+'(\attackthread)$, so {\bf Inv-3} holds for $k+1$.
				\end{asparaitem}
	
	\item[\fbox{2}\,``$\anevent$ is not a flush event, $\pcconf'(\attackthread)=\extensionstate{i}$ for $i\in{\intrange{1}{n-1}}$, 
		$\instructionOf{\anevent}\neq\aninstruction_{i+1}$.''] \hfill Event $\anevent$ corresponds to an instruction 
			that does not follow $\reowit$. Then, events for instructions (\ref{Flush:NextFirst}--\ref{Another:Inst})
			place the auxiliary address-value pairs into $\bufconf{}_+'(\attackthread)$ and then perform $\commandOf{\instructionOf{\anevent}}$.  
			Let $\anevent_{1},\ldots,\anevent_{\iterator},\anevent'\in\eventsextended$ and $\astate'_{k+1}\in\statesextended$ 
			such that $\anevent_{j}$ are, for all $j\in\intrange{1}{\iterator}$, the buffering events for stores
			(\ref{Flush:NextFirst}--\ref{Flush:NextLast}), $\anevent'$ is the event for instruction (\ref{Another:Inst}), 
			and ${\astate'_k}\transitionto{\anevent_{1}\cdot\ldots\cdot\anevent_{\iterator}\cdot\anevent'}{\astate'_{k+1}}\in\eventtransitions{\tso}^*$,
			according to the Figure~\ref{Figure:TSORules} rules.
			We define $\bcomp_{k+1}:=\bcomp_k\cdot\anevent_{1}\cdot\ldots\cdot\anevent_{\iterator}\cdot\anevent'$ 
			and find that invariants {\bf Inv-0...3} hold for $k+1$ since
				\begin{asparaitem}
					\item[\quad(0)] {\bf Inv-0} holds for $k$ so $\astate_0\transitionto{\acomp_{k+1}}\astate_{k+1}$ 
						and $\astate'_0\transitionto{\bcomp_{k+1}}\astate'_{k+1}$, i.e. {\bf Inv-0} holds for $k+1$.
					\item[\quad(1)] {\bf Inv-1} holds for $k$ and 
						$\pcconf{}_+(\attackthread)=\destinationstateOf{\instructionOf{\anevent}}=\pcconf{}_+'(\attackthread)$, 
						so {\bf Inv-1} holds for $k+1$.
					\item[\quad(2)] {\bf Inv-2} holds for $k$ and the events $\anevent$ and $\anevent'$ update at most one $\regdomain$ register 
						by the same value, so {\bf Inv-2} holds for $k+1$.
					\item[\quad(3)] {\bf Inv-3} holds for $k$, the buffering store events $\anevent_{1},\ldots,\anevent_{\iterator}$ 
						make the address-value pairs of the auxiliary registers explicit in $\bufconf{}_+'(\attackthread)$, 
						and if events $\anevent$ and $\anevent'$ are buffering events for stores then they add the same address-value pair, 
						so {\bf Inv-3} holds for $k+1$.
				\end{asparaitem}

	\item[\fbox{3}\,``$\instructionOf{\anevent}$ performs a store and \fbox{2} fails.''] %
		We analyze the following subcases depending on the value of $\pcconf'(\attackthread)$.
		\begin{asparaitem}
			\item[\fbox{3a}\,``$\pcconf'(\attackthread)=\extensionstate{i-1}$ for some $i\in{\intrange{1}{n-1}}$.''] 
				Since \fbox{2} does not hold, $\instructionOf{\anevent}=\aninstruction_i$ 
				and auxiliary registers track the store $\aninstruction_i$. 
				Let $\anevent_\anaddr,\anevent_\aval\in\eventsextended$ be events for the instructions in (\ref{Store:AddressValue}) 
				and $\astate'_{k+1}\in\statesextended$ such that
				$\astate'_k\transitionto{\anevent_\anaddr\cdot\anevent_\aval}\astate'_{k+1}\in\eventtransitions{\tso}^*$ according to
				the $\eventtransitions{\tso}$ rule for local assignments.
				We define $\bcomp_{k+1}:=\bcomp_k\cdot\anevent_\anaddr\cdot\anevent_\aval$ and
				find that invariants {\bf Inv-0...3} hold for $k+1$ since
					\begin{asparaitem}
						\item[\quad(0)] {\bf Inv-0} holds for $k$ so $\astate_0\transitionto{\acomp_{k+1}}\astate_{k+1}$
							and $\astate'_0\transitionto{\bcomp_{k+1}}\astate'_{k+1}$, i.e. {\bf Inv-0} holds for $k+1$.
						\item[\quad(1)] {\bf Inv-1} holds for $k$, $\pcconf{}_+(\attackthread)=\destinationstateOf{\aninstruction_i}$, 
							and $\pcconf{}_+'(\attackthread)=\extensionstate{i}$, so {\bf Inv-1} holds for $k+1$.
						\item[\quad(2)] {\bf Inv-2} holds for $k$ and no memory changes occurred outside of auxiliary registers, so {\bf Inv-2} holds for $k+1$.
						\item[\quad(3)] {\bf Inv-3} holds for $k$ and $({\valOf{\adrreg_\iterator}},{\valOf{\valreg_\iterator}})$ matches 
							the address-value pair added by $\anevent$ to $\bufconf{}_+(\attackthread)$, so {\bf Inv-3} holds for $k+1$.
					\end{asparaitem}

			\item[\fbox{3b}\,``$\pcconf'(\attackthread)=\pcconf(\attackthread)\neq\sourcestateOf{\aninstruction_1}$.'']
				This case is sim\-i\-lar to the one when $\threadOf{\anevent}\neq\attackthread$ since $\instructionOf{\anevent}\in\instructionsextended$.
				Then there exist $\anevent'\in\eventsextended$ and $\astate'_{k+1}\in\statesextended$ such that 
				$\instructionOf{\anevent'}=\instructionOf{\anevent}$ and 
				$\atransition{\astate'_k}{\astate'_{k+1}}{\anevent'}\in\eventtransitions{\tso}$ in $\eventautomaton{\tso}{\bprog\glue\reowit}$.
				We define $\bcomp_{k+1}:=\bcomp_k\cdot\anevent'$ and find that, by the $\eventtransitions{\tso}$ semantics (Figure~\ref{Figure:TSORules}), 
				invariants {\bf Inv-0...3} continue to hold for $k+1$.
		\end{asparaitem}

	\item[\fbox{4}\,``$\instructionOf{\anevent}$ performs a load and \fbox{2} fails.''] %
		We analyze the following subcases depending on the value of $\pcconf'(\attackthread)$.
		\begin{asparaitem}
			\item[\fbox{4a}\,``$\pcconf'(\attackthread)=\extensionstate{i-1}$ for some $i\in{\intrange{1}{n-1}}$.'']
				Since \fbox{2} does not hold, 
				$\instructionOf{\anevent}=\aninstruction_i$ and we use (\ref{Load:BufferFirst}--\ref{Load:BufferLast},\ref{Load:Memory}) 
				to load from $\anexpr$ only when no register $\adrreg_j$ matches $\anexpr$ for any $j\in\intrange{1}{\iterator}$.

				\quad If there exists a largest $j\in\intrange{1}{\iterator}$ such that $\adrreg_j=\anexpr$ then 
					$\areg$ will take its value from the auxiliary register $\valreg_j$. 
					Let $\anevent_\iterator,\ldots,\anevent_j,\anevent_\assignevent\in\eventsextended$ and $\astate'_{k+1}\in\statesextended$ such that 
					$\anevent_k$ are, for all $k\in\intrange{j+1}{\iterator}$, 
					the events for negative conditional checks (\ref{Load:CheckNegativeFirst},\ref{Load:CheckNegativeLast}),
					$\anevent_j$ is the event for the earliest positive conditional check (\ref{Load:CheckPositiveFirst},\ref{Load:CheckPositiveLast}),
					$\anevent_\assignevent$ is the event for an instruction (\ref{Load:BufferFirst},\ref{Load:BufferLast}), and 
					$\astate'_k\transitionto{\anevent_\iterator\cdot\ldots\cdot\anevent_j\cdot\anevent_\assignevent}\astate'_{k+1}\in\eventtransitions{\tso}^*$
					according to the rules for conditionals and local assignments in $\eventtransitions{\tso}$.
					We define $\bcomp_{k+1}:=\bcomp_k\cdot\anevent_\iterator\cdot\ldots\cdot\anevent_j\cdot\anevent_\assignevent$ and find that
					the invariants {\bf Inv-0...3} hold for $k+1$ since
					\begin{asparaitem}
						\item[(0)] {\bf Inv-0} holds for $k$ so $\astate_0\transitionto{\acomp_{k+1}}\astate_{k+1}$
							and $\astate'_0\transitionto{\bcomp_{k+1}}\astate'_{k+1}$, i.e. {\bf Inv-0} holds for $k+1$.
						\item[(1)] {\bf Inv-1} holds for $k$, $\pcconf{}_+(\attackthread)=\destinationstateOf{\aninstruction_i}$, 
							and $\pcconf{}_+'(\attackthread)=\extensionstate{i}$, so {\bf Inv-1} holds for $k+1$.
						\item[(2)] {\bf Inv-2} holds for $k$, both $\anevent$ and $\anevent_\assignevent$ update $\areg$ by the same value, 
							and no other event $\anevent_\iterator,\ldots,\anevent_j$ changes any address, so {\bf Inv-2} holds for $k+1$.
						\item[(3)] {\bf Inv-3} holds for $k$ and no event alters buffer contents, so {\bf Inv-3} holds for $k+1$.
					\end{asparaitem}
		
				\quad Otherwise, $\adrreg_j\neq\anexpr$ holds for all $j\in\intrange{1}{\iterator}$ 
					and the register $\areg$ will take its value from the address indicated by $\anexpr$.
					Namely, let $\anevent_\iterator,\ldots,\anevent_1,\anevent_\loadevent\in\eventsextended$ and $\astate'_{k+1}\in\statesextended$ such that
					$\anevent_k$ are, for all $k\in\intrange{1}{\iterator}$, 
					the events for negative conditional checks (\ref{Load:CheckNegativeFirst},\ref{Load:CheckNegativeLast}),
					$\anevent_\loadevent$ is the event for instruction (\ref{Load:Memory}), and 
					$\astate'_k\transitionto{\anevent_\iterator\cdot\ldots\cdot\anevent_1\cdot\anevent_\loadevent}\astate'_{k+1}\in\eventtransitions{\tso}^*$
					according to the rule for conditionals in $\eventtransitions{\tso}$ and ({LB}/{LM}).
					We define $\bcomp_{k+1}:=\bcomp_k\cdot\anevent_\iterator\cdot\ldots\cdot\anevent_1\cdot\anevent_\loadevent$ and find that
					invariants {\bf Inv-0...3} hold for $k+1$:
					\begin{asparaitem}
						\item[(0)] {\bf Inv-0} holds for $k$ so $\astate_0\transitionto{\acomp_{k+1}}\astate_{k+1}$
							and $\astate'_0\transitionto{\bcomp_{k+1}}\astate'_{k+1}$, i.e. {\bf Inv-0} holds for $k+1$.
						\item[(1)] {\bf Inv-1} holds for $k$, $\pcconf{}_+(\attackthread)=\destinationstateOf{\aninstruction_i}$, 
							and $\pcconf{}_+'(\attackthread)=\extensionstate{i}$, so {\bf Inv-1} holds for $k+1$.
						\item[(2)] {\bf Inv-2} holds for $k$, both $\anevent$ and $\anevent_\loadevent$ update $\areg$ by the same value, 
							and no other event $\anevent_\iterator,\ldots,\anevent_1$ changes any address, so {\bf Inv-2} holds for $k+1$.
						\item[(3)] {\bf Inv-3} holds for $k$ and no event alters buffer contents, so {\bf Inv-3} holds for $k+1$.
					\end{asparaitem}
			\item[\fbox{4b}\,``$\pcconf'(\attackthread)=\extensionstate{n-1}$.''] 
				Since \fbox{2} does not hold, $\instructionOf{\anevent}=\aninstruction_n$. Furthermore, because $\iterator=\bound$,
				additionally to performing the events that simulate the load behavior as in subcase \fbox{4a}, 
				the extension returns to the original program flow using events for (\ref{Flush:First}--\ref{Flush:Last}) 
				and makes the auxiliary registers address-value pairs explicit in $\bufconf{}_+'(\attackthread)$.
				
				\quad Let $\anevent_1',\ldots,\anevent_\bound'\in\eventsextended$ and $\astate'_{k+1}\in\statesextended$ such that
				$\anevent_k'$ are, for all $k\in\intrange{1}{\bound}$, the buffering events for stores (\ref{Flush:First},\ref{Flush:Last}),
				and $\astate''_{k+1}\transitionto{\anevent_1'\cdot\ldots\cdot\anevent_\bound'}\astate'_{k+1}\in\eventtransitions{\tso}^*$
				according to ({LS}) from Figure~\ref{Figure:TSORules}, with $\astate''_{k+1}$ being notation for $\astate'_{k+1}$ from \fbox{4a}. 
				We define $\bcomp_{k+1}:=\bcomp'_{k+1}\cdot\anevent_1'\cdot\ldots\cdot\anevent_\bound'$, where $\bcomp'_{k+1}$ is 
				notation for $\bcomp_{k+1}$ from \fbox{4a}, and find that the invariants {\bf Inv-0...3} hold for $k+1$ since
				\begin{asparaitem}
					\item[(0)] {\bf Inv-0} holds for $k$ so $\astate_0\transitionto{\acomp_{k+1}}\astate_{k+1}$
						and $\astate'_0\transitionto{\bcomp_{k+1}}\astate'_{k+1}$, i.e. {\bf Inv-0} holds for $k+1$.
					\item[(1)] {\bf Inv-1} holds for $k$ and $\pcconf{}_+(\attackthread)=\destinationstateOf{\aninstruction_n}=\pcconf{}_+'(\attackthread)$, 
						so {\bf Inv-1} holds for $k+1$.
					\item[(2)] {\bf Inv-2} holds for $k$, both events $\anevent$ and $\anevent_\loadevent$ update $\areg$ by the same value, 
						and no other event $\anevent_\iterator,\ldots,\anevent_1,\anevent_1',\ldots,\anevent_\bound'$ changes any address, 
						so {\bf Inv-2} holds for $k+1$. 
					\item[(3)] {\bf Inv-3} holds for $k$ and events $\anevent_1',\ldots,\anevent_\bound'$ 
						place the corresponding address-value pairs that match $\bufconf{}_+(\attackthread)$ into $\bufconf{}_+'(\attackthread)$, 
						so {\bf Inv-3} holds for $k+1$.
				\end{asparaitem}
			\item[\fbox{4c}\,``$\pcconf'(\attackthread)=\pcconf(\attackthread)$.'']
				This case is similar to \fbox{3b}.
				Let $\anevent'\in\eventsextended$ and $\astate'_{k+1}\in\statesextended$ such that 
				$\instructionOf{\anevent'}=\instructionOf{\anevent}$ and 
				$\atransition{\astate'_k}{\astate'_{k+1}}{\anevent'}\in\eventtransitions{\tso}$ in $\eventautomaton{\tso}{\bprog\glue\reowit}$.
				We define $\bcomp_{k+1}:=\bcomp_k\cdot\anevent'$ and find that, by the $\eventtransitions{\tso}$ semantics (Figure~\ref{Figure:TSORules}), 
				the invariants {\bf Inv-0...3} hold for $k+1$.
		\end{asparaitem}

	\item[\fbox{5}\,``$\anevent$ performs an assignment, conditional, or memory fence and \fbox{2} fails.''] %
		We analyze the following subcases.
		\begin{asparaitem}
			\item[\fbox{5a}\,``$\pcconf'(\attackthread)=\extensionstate{i-1}$ for $i\in{\intrange{1}{n-1}}$.''] 
				Since \fbox{2} does not hold, $\instructionOf{\anevent}=\aninstruction_i$ is either a conditional or an assignment.

				\quad If $\commandOf{\aninstruction_i}=\thelocal{\areg}{\anexpr}$ 
					let $\anevent'\in\eventsextended$ and $\astate'_{k+1}\in\statesextended$ such that
					$\instructionOf{\anevent'}=\atransition{\acontrolstate_{i-1}}{\acontrolstate_i}{\thelocal{\areg}{\anexpr}}$ and 
					$\atransition{\astate'_k}{\astate'_{k+1}}{\anevent'}\in\eventtransitions{\tso}$ by the $\eventtransitions{\tso}$ rule for local assignments.
					We define $\bcomp_{k+1}:=\bcomp_k\cdot\anevent'$ and find that the invariants {\bf Inv-0...3} hold for $k+1$ since
					\begin{asparaitem}
						\item[(0)] {\bf Inv-0} holds for $k$ so $\astate_0\transitionto{\acomp_{k+1}}\astate_{k+1}$
							and $\astate'_0\transitionto{\bcomp_{k+1}}\astate'_{k+1}$, i.e. {\bf Inv-0} holds for $k+1$.
						\item[(1)] {\bf Inv-1} holds for $k$, $\pcconf{}_+(\attackthread)=\destinationstateOf{\aninstruction_i}$, 
							and $\pcconf{}_+'(\attackthread)=\extensionstate{i}$, so {\bf Inv-1} holds for $k+1$.
						\item[(2)] {\bf Inv-2} holds for $k$ and $\anexpr$ is evaluated the same by both $\anevent$ and $\anevent'$,
							so the register $\areg$ is updated by the same value and {\bf Inv-2} holds for $k+1$.
						\item[(3)] {\bf Inv-3} holds for $k$ and no event alters buffer contents, so {\bf Inv-3} holds for $k+1$.
					\end{asparaitem}

				\quad Otherwise, $\commandOf{\aninstruction_i}=\theassume{\anexpr}$. 
					Let $\anevent'\in\eventsextended$ and $\astate'_{k+1}\in\statesextended$
					such that $\instructionOf{\anevent'}=\atransition{\acontrolstate_{i-1}}{\acontrolstate_i}{\theassume{\anexpr}}$ and 
					$\atransition{\astate'_k}{\astate'_{k+1}}{\anevent'}\in\eventtransitions{\tso}$ by the $\eventtransitions{\tso}$ rule for conditionals.
					We define $\bcomp_{k+1}:=\bcomp_k\cdot\anevent'$ and find that the invariants {\bf Inv-0...3} hold for $k+1$ since
				\begin{asparaitem}
					\item[(0)] {\bf Inv-0} holds for $k$ so $\astate_0\transitionto{\acomp_{k+1}}\astate_{k+1}$
						and $\astate'_0\transitionto{\bcomp_{k+1}}\astate'_{k+1}$, i.e. {\bf Inv-0} holds for $k+1$.
					\item[(1)] {\bf Inv-1} holds for $k$, $\pcconf{}_+(\attackthread)=\destinationstateOf{\aninstruction_i}$, 
						and $\pcconf{}_+'(\attackthread)=\extensionstate{i}$, so {\bf Inv-1} holds for $k+1$.
					\item[(2)] {\bf Inv-2} holds for $k$ and both $\anevent$ and $\anevent'$ do not change any address,
						so {\bf Inv-2} holds for $k+1$.
					\item[(3)] {\bf Inv-3} holds for $k$ and no event alters buffer contents, so {\bf Inv-3} holds for $k+1$.
				\end{asparaitem}

			\item[\fbox{5b}\,``$\pcconf'(\attackthread)=\pcconf(\attackthread)$.'']
				This case covers the remaining possibilities when $\anevent$ is an assignment, conditional, or memory fence.
				Similar to cases \fbox{3b} and \fbox{4c}, let $\anevent'\in\eventsextended$ and $\astate'_{k+1}\in\statesextended$ such that 
				$\instructionOf{\anevent'}=\instructionOf{\anevent}$ and 
				$\atransition{\astate'_k}{\astate'_{k+1}}{\anevent'}\in\eventtransitions{\tso}$ in $\eventautomaton{\tso}{\bprog\glue\reowit}$.
				We define $\bcomp_{k+1}:=\bcomp_k\cdot\anevent'$ and find that, by the $\eventtransitions{\tso}$ semantics (Figure~\ref{Figure:TSORules}), 
				invariants {\bf Inv-0...3} hold for $k+1$.	
		\end{asparaitem}
\end{asparaitem}

The above case distinction covers all possibilities for events $\anevent$ that $\acomp$ may perform from $\astate_k$.
Hence, by complete induction, the extension does not remove \tso{}-reachable states:
if $\astate=(\pcconf,\valconf,\bufconf)$ is reachable by $\acomp$
then there exists $\astate'=(\pcconf',\valconf',\bufconf')$ and $\bcomp\in\computationsOf{\tso}{\bprog\glue\reowit}$
such that $\astate'$ is reachable by $\bcomp$ in $\bprog\glue\reowit$, $\pcconf=\pcconf'$, $\valconf(\anaddr)=\valconf'(\anaddr)$
for all $\anaddr\in\addrdomain\cup\regdomain$, and $\bufconf=\bufconf'$ are empty.

For the reverse direction, let $\proj{\awitness}\colon\computationsOf{\tso}{\bprog}\to\computationsOf{\tso}{\bprog\glue\awitness}$ 
be the map $\acomp\mapsto\bcomp$ that the inductive proof implies,
respectively $\proj{\awitness}\colon\events\to\eventsextended^*$ its restriction to events matching the different inductive cases.
Furthermore, consider computations $\bcomp\in\computationsOf{\tso}{\bprog\glue\reowit}$
that do not interleave events of other threads within the events of sequences $\proj{\awitness}(\anevent)$.
Such computations reach the entire set $\reachOf{\tso}{\bprog\glue\reowit}$.
E.g.,\ since local events $\anevent_\iterator,\ldots,\anevent_1$ as in case \fbox{4a} that precede $\anevent_\loadevent$
can be performed right before $\anevent_\loadevent$, 
the above restriction does not change the set of \tso{}-reachable states in $\bprog\glue\reowit$.
Note that $\proj{\awitness}$ is a bijection between such computations $\bcomp$ and computations $\acomp\in\computationsOf{\tso}{\bprog}$ 
that delay flushes locally the least wrt. $\attackthread$.
Another induction can show that for each computation $\bcomp$ as described above
there exists a computation $\acomp\in\computationsOf{\tso}{\bprog}$ such that invariants {\bf Inv-0...3} hold for prefixes of $\bcomp$ and $\acomp$.
This implies that the extension by $\reowit$ does not add \tso{}-reachable states. \qed
\endproofatend

We now show that the oracle never suggests the same sequence $\sigma$ twice.
Since in $\bprog\glue\reowit$ we introduce new instructions that correspond to instructions in $\bprog$, we have to map back sequences of instructions $\instructionsextended$ in $\bprog\glue\reowit$ to sequences of instructions  $\instructions$ in $\bprog$.
Intuitively, the mapping gives the original instructions from which the sequence was produced.
Formally, we define a family of projection functions $\wsproj{\reowit}\colon\instructionsextended^*\to\instructions^*$ with $\wsproj{\reowit}(\varepsilon):=\varepsilon$
and $\wsproj{\reowit}(w\cdot\aninstruction):=\wsproj{\reowit}(w)\cdot\wsproj{\reowit}(\aninstruction)$.
For an instruction $\aninstruction\in\instructionsextended$, we define
$\wsproj{\reowit}(\aninstruction):=\aninstruction$ provided $\aninstruction\in\instructions$. 
We set $\wsproj{\reowit}(\aninstruction):=\aninstruction_i$
if $\aninstruction$ is a first instruction on the path between $\extensionstate{i-1}$ and $\extensionstate{i}$ for some $i\in\intrange{1}{n}$.
In all other cases, we delete the instruction, $\wsproj{\reowit}(\aninstruction):=\varepsilon$.
Then, if $\bprog_0:=\aprog$ is the original program,
$\reowit_j$ is the sequence that the oracle returns in iteration $j\in\nat$ of the while loop,
and $w$ is a sequence of instructions in $\bprog_{j+1}$,
we define $\wsproj{}(w):=\wsproj{\reowit_0}(\ldots\wsproj{\reowit_j}(w))$.
This latter function maps sequences of instructions in program $\bprog_{j+1}$ back to sequences of instructions in $\aprog$.

We are ready to state our key lemma.
Intuitively, if the oracle in Algorithm~\ref{Algorithm:Reachability} returns $\reowit:=\oracleOf{\bprog}$ and $\reowit':=\oracleOf{\bprog\glue\reowit}$
then, necessarily, $\wsproj{}(\reowit')\neq\wsproj{}(\reowit)$.

\begin{alemma}\label{Lemma:NewWitnessSupport}%
Let $\bprog_0:=\aprog$ and $\bprog_{i+1}:=\bprog_i\glue\reowitof{i}$ for $\reowitof{i}:=\oracleOf{\bprog_i}$ as in Algorithm~\ref{Algorithm:Reachability}.
If $\reowitof{j+1}\neq\varepsilon$ then $\wsproj{}(\reowitof{j+1})\neq\wsproj{}(\reowitof{i})$ for all $i\leq j$.
\end{alemma}
\begin{proof}
Assume, to the contrary, that $\wsproj{}(\reowitof{j+1})=\wsproj{}(\reowitof{i})$ for some $i\leq j$ where
$\reowitof{j+1}:=\oracleOf{\bprog_{j+1}}$ and $\reowitof{i}:=\oracleOf{\bprog_i}$.
Let $\aninstruction_\text{first}$ be the first (store) instruction and $\aninstruction_\text{last}$ the last (load) instruction of $\reowitof{j+1}$.
Similarly, let $\aninstruction_\text{first}'$ and $\aninstruction_\text{last}'$ be the first and last instructions of $\reowitof{i}$.
Since $\wsproj{}(\reowitof{j+1})=\wsproj{}(\reowitof{i})$ it means that
$\wsproj{}(\aninstruction_\text{first})=\wsproj{}(\aninstruction_\text{first}')$ and
$\wsproj{}(\aninstruction_\text{last})=\wsproj{}(\aninstruction_\text{last}')$.

However,
since all control flows of $\bprog_{i+1}:=\bprog_i\glue\reowitof{i}$ that recreate $\wsproj{}(\aninstruction_\text{first}')$ before
$\wsproj{}(\aninstruction_\text{last}')$ also place a fence between the two,
no other later sequences that the oracle returns have $\wsproj{}(\aninstruction_\text{first}')$ come before $\wsproj{}(\aninstruction_\text{last}')$.
This in particular means that $\reowitof{j+1}=\oracleOf{\bprog_{j+1}}$
where $\wsproj{}(\aninstruction_\text{first})$ comes before $\wsproj{}(\aninstruction_\text{last})$ does not exist.
In conclusion, the initial assumption is false. \qed
\end{proof}

We can now prove Algorithm~\ref{Algorithm:Reachability} sound and complete for acyclic programs (Theorem~\ref{Theorem:Correctness}).
Lemma~\ref{Lemma:NewWitnessSupport} and the assumption that the input program is acyclic ensure that
if no goal state is found \seqcon{}-reachable (Line~\ref{Algorithm:Reachable}), then Algorithm~\ref{Algorithm:Reachability}
eventually runs out of sequences $\reowit$ to return (Line~\ref{Algorithm:RobustnessCheck}).
If that is the case, $\oracleOf{\bprog}$ returns $\varepsilon$ in the last iteration of Algorithm~\ref{Algorithm:Reachability}.
By the first oracle condition, we know that the \seqcon{}- and \tso{}-reachable states of $\bprog$ are the same.
Hence, no goal state is \tso{}-reachable in $\bprog$ and, by Lemma~\ref{Lemma:Extension},
no goal state is \tso{}-reachable in the input program $\aprog$ either.
Otherwise, a goal state $\astate$ is \seqcon{}-reachable by some computation $\tau$ in $\bprog_j$ for some $j\in\nat$
and, by Lemma~\ref{Lemma:Extension}, there is a \tso{} computation in $\aprog$ corresponding to $\tau$ that reaches $\astate$.

\begin{atheorem}\label{Theorem:Correctness}
For acyclic programs, Algorithm~\ref{Algorithm:Reachability} terminates. Moreover,
it returns \mytrue{} on input $\aprogram$ if and only if $\reachOf{\tso}{\aprogram}\cap\acceptingstates\neq\emptyset$.
\end{atheorem}
\begin{proof}
It is immediate that Algorithm~\ref{Algorithm:Reachability} always terminates for acyclic programs.
On the one hand, the number of instruction sequences that start with a store and end with a load as in the second oracle condition are finite in $\aprogram$.
On the other hand, by Lemma~\ref{Lemma:NewWitnessSupport},
at each iteration the oracle returns a sequence that differs (in $\aprogram$) from the previous ones.
These two facts imply termination.

We now prove that
$\reachOf{\tso}{\aprogram}\cap\acceptingstates\neq\emptyset$ iff. Algorithm~\ref{Algorithm:Reachability} returns \mytrue\ on input $\aprogram$.
For the easy direction, assume that Algorithm~\ref{Algorithm:Reachability} returns \mytrue\ on input $\aprogram$.
This means that $\reachOf{\seqcon}{\bprog}\cap\acceptingstates\neq\emptyset$ in the last iteration of the algorithm's loop.
Then, by $\reachOf{\seqcon}{\bprog}\subseteq\reachOf{\tso}{\bprog}$ and Lemma~\ref{Lemma:Extension},
we know that $\reachOf{\seqcon}{\bprog}\subseteq\reachOf{\tso}{\aprogram}$.
Hence, $\reachOf{\tso}{\aprogram}\cap\acceptingstates\neq\emptyset$.

For the reverse direction, assume that $\reachOf{\tso}{\aprogram}\cap\acceptingstates\neq\emptyset$.
Furthermore, let $\bprog_0:=\aprogram$ and $\bprog_{i+1}:=\bprog_i\glue\reowitof{i}$ for $\reowitof{i}:=\oracleOf{\bprog_i}$.
By the initial termination argument we know there exists $j\in\nat$ such that the algorithm terminates with $\bprog=\bprog_{j}$ in its last loop iteration.
That means that either the check in Line~\ref{Algorithm:SCReachabilityCheck} of the algorithm succeeds,
in which case Algorithm~\ref{Algorithm:Reachability} returns \mytrue,
or the check in Line~\ref{Algorithm:RobustnessCheck} of the algorithm fails,
i.e.\ $\oracleOf{\bprog_j}=\epsilon$ and $\reachOf{\seqcon}{\bprog_j}\cap\acceptingstates=\emptyset$.
In the latter case, by the first oracle condition we know that $\reachOf{\tso}{\bprog_j}\cap\acceptingstates=\emptyset$ and,
by Lemma~\ref{Lemma:Extension}, we get $\reachOf{\tso}{\bprog_j}\subseteq\reachOf{\tso}{\bprog_0}$.
Then, $\reachOf{\tso}{\aprogram}\cap\acceptingstates=\emptyset$ contradicts the above assumption and concludes the proof. \qed
\end{proof}

To establish that Algorithm~\ref{Algorithm:Reachability} is a semi-decision procedure for all programs,
one can use an iterative bounded model checking approach.
Bounded model checking unrolls the input program $\aprogram$ up to a bound $k\in\nat$ on the length of computations.
Then Algorithm~\ref{Algorithm:Reachability} is applied to the resulting programs $\aprog_k$.
If it finds a goal state \tso{}-reachable in $\aprog_k$, %
this state corresponds to a \tso{}-reachable goal state in $\aprogram$.
Otherwise, we increase $k$ and try again.
By Theorem~\ref{Theorem:Correctness}, we know that Algorithm~\ref{Algorithm:Reachability} is a decision procedure for each $\aprogram_k$.
This implies that Algorithm~\ref{Algorithm:Reachability} together with iterative bounded model checking yields a semi-decision procedure
that terminates for all positive instances of \tso{} reachability.
For negative instances of \tso{} reachability, however,
the procedure is guaranteed to terminate only if the input program $\aprogram$ is acyclic. %

\begin{atheorem}\label{Theorem:PartialCorrectness}
We have $\acceptingstates\cap\reachOf{\tso}{\aprogram}\neq\emptyset$ if and only if, for large enough $k\in\nat$,
Algorithm~\ref{Algorithm:Reachability} returns \mytrue{} on input $\aprogram_k$.
\end{atheorem}
\begin{proof}
Assume that $\acceptingstates\cap\reachOf{\tso}{\aprogram}\neq\emptyset$.
Then there exist some state $\astate\in\acceptingstates$ and $\acomp\in\computationsOf{\tso}{\aprogram}$
such that $\initialeventstate\transitionto{\acomp}\astate$.
Let $k$ be the length of $\acomp$ and $\acceptingstates'$ be the goal states of $\eventautomaton{\tso}{\aprogram_k}$.
There exists a computation $\bcomp\in\computationsOf{\tso}{\aprogram_k}$ that mimics $\acomp$ and reaches $\astate'\in\acceptingstates'$.
Hence, $\acceptingstates'\cap\reachOf{\tso}{\aprogram_k}\neq\emptyset$ and, by Theorem~\ref{Theorem:Correctness},
Algorithm~\ref{Algorithm:Reachability} returns \mytrue{} on input $\aprogram_k$.

For the reverse direction, assume that Algorithm~\ref{Algorithm:Reachability} returns \mytrue{} on input $\aprogram_k$ for some $k\in\nat$.
Let $\initialeventstate'$ be the initial state of $\eventautomaton{\tso}{\aprogram_k}$ and, as before,
$\acceptingstates'$ be the goal states of $\eventautomaton{\tso}{\aprogram_k}$.
By Theorem~\ref{Theorem:Correctness}, there exists $\astate'\in\acceptingstates'\cap\reachOf{\tso}{\aprogram_k}$
and $\bcomp\in\computationsOf{\tso}{\aprogram_k}$ such that $\initialeventstate'\transitionto{\bcomp}\astate'$.
Since $\aprogram_k$ unrolls $\aprogram$ up to bound $k$,
there exists a computation $\acomp\in\computationsOf{\tso}{\aprogram}$ that mimics $\bcomp$ and reaches $\astate\in\acceptingstates$.
Therefore, $\acceptingstates\cap\reachOf{\tso}{\aprogram}\neq\emptyset$. \qed
\end{proof}
\section{A Robustness-based Oracle}\label{Section:Robustness} %
This section argues why robustness yields an oracle.
Robustness~\cite{AlglaveM11,BDM13,Sen2011,ShashaSnir88} is a correctness criterion requiring that for each \tso\ computation of a program
there is an \seqcon{} computation that has the same data and control dependencies.
Delays due to store buffering are still allowed,
as long as they do not produce dependencies between instructions that \seqcon{} computations forbid.

Dependencies between events are described in terms of the \emph{happens-before} relation of a computation $\tau\in\computationsOf{\tso}{\aprogram}$.
The happens-before relation is a union of the three relations that we define below:
$\happensbeforeOf{\tau} :=\ \progorder\cup\equivalenceorder\cup\conflictorder$.

The \emph{program order relation} $\progorder$ is the order in which threads issue their commands.
Formally, it is the union of the program order relations for all threads:
$\progorder\ :=\ \bigcup_{\athread\in\threaddomain}\progorder^{\athread}$.
Let $\tau'$ be the subsequence of all non-flush events of thread $\athread$ in $\tau$. Then $\progorder^{\athread}\ :=\,\succorder{\tau'}$.

The \emph{equivalence relation} $\equivalenceorder$ links, in each thread, flush events and their matching store events:
$(\athread,\aninstruction,\anaddr)\equivalenceorder(\athread,\writekind,\anaddr)$.

The \emph{conflict relation} $\conflictorder$ orders accesses to the same address.
Assume, on the one hand, that
$\tau=\tau_1\cdot\storeevent\cdot\tau_2\cdot\loadevent\cdot\tau_3\cdot\flushevent\cdot\tau_4$ such that
$\storeevent\equivalenceorder\flushevent$,
events $\storeevent$ and $\loadevent$ access the same address $\anaddr$ and come from thread $\athread$,
and there is no other store event $\storeevent'\in\tau_2$ such that
$\threadOf{\storeevent'}=\athread$ and $\addrOf{\storeevent'}=\anaddr$.
Then the load event $\loadevent$ is an \emph{early read} of the value buffered by the event $\storeevent$
and $\storeevent\conflictorder\loadevent$.

On the other hand, assume $\tau=\tau_1\cdot\anevent\cdot\tau_2\cdot\anevent'\cdot\tau_3$ such that
$\anevent$ and $\anevent'$ are either load or flush events that access the same address $\anaddr$,
neither $\anevent$ nor $\anevent'$ is an early read,
and at least one of $\anevent$ or $\anevent'$ is a flush to $\anaddr$.
If there is no other flush event $\flushevent\in\tau_2$ with $\addrOf{\flushevent}=\anaddr$ then $\anevent\conflictorder\anevent'$.

Figure~\ref{Figure:DekkerTrace} depicts the happens-before relation of computation $\tau_\wit$.

\begin{wrapfigure}{r}{0.395\textwidth}
	\vskip -1.25em
	\small\centering
	\begin{tikzpicture}[nodes={rectangle,draw=none,fill=none}, node distance=.65cm]
	   \node (store0) {$\one{\storeevent}$};
	   \node [node distance=2.5cm,right of=store0] (store1) {$\two{\storeevent}$};
	   \node [below of=store0] 	(flush0) {$\one{\flushevent}$};
	   \node [below of=store1] 	(flush1) {$\two{\flushevent}$};
	   \node [below of=flush0] 	(load0) {$\one{\loadevent}$}; 
		\node [below of=flush1]	(load1) {$\two{\loadevent}$};
	
		\draw[<->] (store0)	edge node [midway,right] {$\scriptstyle\equivalence$} (flush0);
		\draw[<->] (store1)	edge node [midway,right] {$\scriptstyle\equivalence$} (flush1);
		\draw[->] (store0)	edge[bend right=40] node [pos=.95,left] {$\scriptstyle \po$} (load0.west);
		\draw[->] (store1)	edge[bend left=40] node [pos=.95,right] {$\scriptstyle \po$} (load1.east);
		\draw[->]  (load0)	edge node [pos=0.85,above] {$\scriptstyle \cf$} (flush1.west);
		\draw[->]  (load1)	edge node [pos=0.85,above] {$\scriptstyle \cf$} (flush0.east);
	\end{tikzpicture}
 	\caption{The relation $\happensbeforeOf{\tau_\wit}$.}
	\vskip -1.25em
	\label{Figure:DekkerTrace}
\end{wrapfigure}

A program $\aprogram$ is said to be \emph{robust} against \tso{} if for each computation $\tau\in\computationsOf{\tso}{\aprogram}$
there exists a computation $\tau'\in\computationsOf{\seqcon}{\aprogram}$ such that $\happensbeforeOf{\tau}=\,\happensbeforeOf{\tau'}$.
If a program $\aprogram$ is robust, then it reaches the same set of final states under \seqcon{} and under \tso{}:
\lemmaatend
\section{TSO Semantics and Proofs missing in Section~\ref{Section:Robustness}}\label{Appendix:TSOSemanticsAndRobustness}
Figure~\ref{Figure:FullTSORules} describes the full \tso\ semantics.
For completeness, states $\astate\in\eventstatesOf{\memmodel}$ 
use the additional event counter $\indexconf\colon\threaddomain\to\nat$ to identify events.
This is used, e.g., to define matching stores and flushes and does not affect in any way our results.

\begin{figure}[!ht]
\vskip -1em
\centering
\therules{
  \hspace{-3em}
  \therule{ % load from buffer
    \acommand\,= \theload{\areg}{\anexpr_\anaddr}, 
    \quad$\anaddr=\valOf{\anexpr_\anaddr}$,
    \quad$\projection{\bufconf(\athread)}{(\nat\times\set{\anaddr}\times\datadomain)} = (\anindex,{\anaddr},{\aval})\cdot\beta$
  }{$\astate\transitionto{(\athread,\indexconf(\athread),\aninstruction,\anaddr)} (\indexconf',\pcconf',\valconf[\areg:=\aval],\bufconf)$}
  {RB}
  \hspace{-3em}
  \therule{ % load from memory
    \acommand\,= \theload{\areg}{\anexpr_\anaddr}, 
    \quad$\anaddr=\valOf{\anexpr_\anaddr}$,
    \quad$\projection{\bufconf(\athread)}{(\nat\times\set{\anaddr}\times\datadomain)} = \varepsilon$,
    \quad$\aval=\valconf(\anaddr)$
  }{$\astate\transitionto{(\athread,\indexconf(\athread),\aninstruction,\anaddr)} (\indexconf',\pcconf',\valconf[\areg:=\aval],\bufconf)$}
  {RM}
  \hspace{-3em}
  \therule{ % store
    \acommand\,= \thestore{\anexpr_\anaddr}{\anexpr_\aval}, 
    \quad$\anaddr=\valOf{\anexpr_\anaddr}$,
    \quad$\aval=\valOf{\anexpr_\aval}$,
    \quad$\anindex=\indexconf(\athread)$
  }{$\astate\transitionto{(\athread,\anindex,\aninstruction,\anaddr)}
    (\indexconf',\pcconf',\valconf,\bufconf[\athread:=(\anindex,\anaddr,\aval)\cdot\bufconf(\athread)])$}
  {LS}
  \hspace{-3.15em}
  \doublerule{ % flush
    $\bufconf(\athread) = \beta\cdot(\anindex,\anaddr,\aval)$
  }{$\astate\transitionto{(\athread,\anindex,\writekind,\anaddr)}(\indexconf,\pcconf,\valconf[\anaddr:=\aval],\bufconf[\athread:=\beta])$}
  { % fence
    \acommand\,= \themfence{}, 
    \quad$\bufconf(\athread)=\varepsilon$
  }{$\astate\transitionto{(\athread,\indexconf(\athread),\aninstruction,\bottom)}(\indexconf',\pcconf',\valconf,\bufconf)$}
  {WM}{LF}
  \hspace{-3em}
  \doublerule{ % assign
    \acommand\,= \thelocal{\areg}{\anexpr},
    \quad$\aval=\valOf{\anexpr}$    
  }{$\astate\transitionto{(\athread,\indexconf(\athread),\aninstruction,\bottom)}(\indexconf',\pcconf',\valconf[\areg:=\aval],\bufconf)$}
  { % assert
    \acommand\,= \theassume{\anexpr}, 
    \quad$\valOf{\anexpr}\neq0$
  }{$\astate\transitionto{(\athread,\indexconf(\athread),\aninstruction,\bottom)}(\indexconf',\pcconf',\valconf,\bufconf)$}
  {LA}{LC}
}
\vskip -1em
\caption{Transition rules for $\eventautomaton{\tso}{\aprogram}$ assuming 
  $\astate=(\indexconf,\pcconf,\valconf,\bufconf)$ with $\pcconf(\athread)=\acontrolstate$,
  $\aninstruction=\acontrolstate\transitionto{\acommand}\acontrolstate'$ in thread $\athread$,
  $\indexconf'=\indexconf[\athread:=\indexconf(\athread)+1]$,
  $\pcconf'=\pcconf[\athread:=\acontrolstate']$.
  We use $\valOf{\anexpr}$ to evaluate $\anexpr$ under $\valconf$ 
  and $\projection{\bufconf(\athread)}{(\nat\times\set{\anaddr}\times\datadomain)}$
  for stores in $\bufconf(\athread)$ that access $\anaddr$.}
\label{Figure:FullTSORules}
\vskip -.5em
\end{figure}

As mentioned in subsection~\ref{Subsection:Semantics}, under \seqcon{}, stores are flushed immediately:
\therules{
	\vspace{-10pt}
	\therule{\acommand\,= \thestore{\anexpr_\anaddr}{\anexpr_\aval},
		\quad$\anaddr=\valOf{\anexpr_\anaddr}$,
		\quad$\aval=\valOf{\anexpr_\aval}$,
		\quad$\anindex=\indexconf(\athread)$}
		{$\astate\transitionto{(\athread,\anindex,\aninstruction,\anaddr)(\athread,\anindex,\writekind,\anaddr)}
			(\indexconf',\pcconf',\valconf[\anaddr:=\aval],\bufconf)$}{LSWM}
}

\begin{alemma}\label{Lemma:SameTraceReach}
If $\acomp,\bcomp\in\computationsOf{\tso}{\aprogram}$,
$\initialeventstate\transitionto{\acomp}\astate$,
and $\happensbeforeOf{\acomp}=\happensbeforeOf{\bcomp}$ then $\initialeventstate\transitionto{\bcomp}\astate$.
\end{alemma}
\begin{proof}
Assume $\initialeventstate\transitionto{\bcomp}\astate'$.
Since $\acomp$ and $\bcomp$ have the same program order $\progorder$,
it means $\astate$ and $\astate'$ have the same index counter $\indexconf$ and program counter $\pcconf$.
Moreover, since $\acomp$ and $\bcomp$ have the same conflict order $\conflictorder$, $\astate$ and $\astate'$ have the same memory valuation $\valconf$.
Finally, since computations $\acomp$ and $\bcomp$ empty the buffers, $\astate$ and $\astate'$ have empty buffers.
In conclusion, $\astate=\astate'$. \qed
\end{proof}
\endlemmaatend
\begin{alemma}\label{Lemma:RobustImpliesReach}
If $\aprogram$ is robust against TSO, then $\reachOf{\seqcon}{\aprogram}=\reachOf{\tso}{\aprogram}$.
\end{alemma}
\begin{proof}
The $\subseteq$ inclusion holds by $\computationsOf{\seqcon}{\aprogram}\subseteq\computationsOf{\tso}{\aprogram}$.
For the reverse, assume that there is a \tso\ computation $\tau\in\computationsOf{\tso}{\aprogram}$
such that $\initialeventstate\transitionto{\tau}\astate$.
Since \aprogram\ is robust, there is an \seqcon\ computation $\tau'\in\computationsOf{\seqcon}{\aprogram}$
such that $\happensbeforeOf{\tau}=\happensbeforeOf{\tau'}$.
Then $\tau'\in\computationsOf{\tso}{\aprogram}$ and, by Lemma~\ref{Lemma:SameTraceReach},
$\initialeventstate\transitionto{\tau'}\astate$ so $\astate$ is \seqcon-reachable. \qed
\end{proof}

Our robustness-based oracle makes use of the following characterization of robustness from earlier work~\cite{BDM13}:
a program $\aprogram$ is not robust against \tso{} iff $\computationsOf{\tso}{\aprogram}$ contains a computation,
called \emph{witness}, as in Figure~\ref{Figure:TSOWitness}. \vskip -1em
\begin{alemma}[\cite{BDM13}]\label{Lemma:RobustEquivNoWitness}
Program $\aprogram$ is robust against TSO if and only if the set of TSO computations $\computationsOf{\tso}{\aprogram}$ contains no witness.
\end{alemma}

A witness $\tau$ delays stores of only one thread in $\aprogram$. The other threads adhere to the \seqcon{} semantics.
Conditions \wita{} -- \witd{} in Figure~\ref{Figure:TSOWitness} describe formally this restrictive behavior.
Furthermore, condition \wite{} implies that no computation $\tau'\in\computationsOf{\seqcon}{\aprogram}$ can satisfy
$\happensbeforeOf{\tau}=\,\happensbeforeOf{\tau'}$.

The computation $\tau_\wit$ is a witness for the program in Figure~\ref{Figure:Dekker}.
Indeed, in no \seqcon{} computation of this program can both loads read the initial values of $\xaddr$ and $\yaddr$.
Relative to Figure~\ref{Figure:TSOWitness}, we have
$\attackstore=\one{\storeevent}$, $\attackload=\one{\loadevent}$, $\attackflush=\one{\flushevent}$,
$\tau_3=\two{\storeevent}\cdot\two{\flushevent}\cdot\two{\loadevent}$, and $\tau_1=\tau_2=\tau_4=\varepsilon$.
\begin{figure}[!ht]
\vskip -0.75em
\centering
\begin{tikzpicture}[node distance=1.25cm]
  \node(inv) {$\tau =$};
  \node(isu) [right of=inv,node distance=1.75cm,color=black]{$\attackstore$};
  \node(ld) [right of=isu,node distance=2cm,color=black] {$\attackload$};
  \node(st) [right of=ld,node distance=2cm,color=black] {$\attackflush$};
  \node(b1) [right of=isu,node distance=1cm] {};
  \node(b2) [right of=st,node distance=1.5cm] {};

  \draw[-, very thick] (inv)edge node[below]{$\tau_1$} (isu);
  \draw[-, very thick] (isu)edge node[below]{$\tau_2$} (ld);
  \draw[-, very thick,color=black] (ld) -- node[below,color=black]{$\tau_3$}(st);
  \draw[-, very thick,color=black] (st)edge node[below,color=black]{$\tau_4$}(b2);
  \draw[-,thick] (isu.north) edge[out=15,in=160] (st.north);
  \draw[-,thick] (b1.north) edge[out=20,in=160] (b2.north);
\end{tikzpicture}
\caption{Witness $\tau$ with $\attackstore\equivalenceorder\attackflush$ and thread $\attackthread:=\threadOf{\attackstore}=\threadOf{\attackload}$.
Witnesses satisfy the following constraints:
\wita~Only thread $\attackthread$ delays stores.
\witb\ Event $\attackflush$ is the first delayed store of $\attackthread$ and
$\attackload$ is the last event of $\attackthread$ past which $\attackflush$ is delayed.
So $\tau_2$ contains neither flush events nor fences of $\attackthread$.
\witc\ Sequence $\tau_3$ contains no events of thread $\attackthread$.
\witd\ Sequence $\tau_4$ consists only of flush events $\anevent$ of thread $\attackthread$.
All these events $\anevent$ satisfy $\addrOf{\anevent}\neq\addrOf{\attackload}$.
\wite\ We require $\attackload\happensbefore^+\anevent$ for all events $\anevent$ in $\tau_3\cdot\attackflush$.}
\vskip -.75em
\label{Figure:TSOWitness}
\end{figure}

The \emph{robustness-based oracle}, given input $\aprogram$,
finds a witness $\awitness$ as in Figure~\ref{Figure:TSOWitness}
and returns the sequence of instructions for the events in $\attackstore\cdot\tau_2\cdot\attackload$
that belong to thread $\attackthread$.
If no witness exists, it returns $\varepsilon$.
By Lemmas~\ref{Lemma:RobustImpliesReach}~and~\ref{Lemma:RobustEquivNoWitness},
this satisfies the oracle conditions from Section~\ref{Section:LazyReachability}.
Note that, given a robust program and the robustness-based oracle as inputs,
Algorithm~\ref{Algorithm:Reachability} returns within the first iteration of the while loop.

\section{Experiments}\label{Section:Experiments}
We have implemented our lazy TSO reachability algorithm on top of the tool \trencher~\cite{Trencher}.
\trencher{} was initially developed for checking robustness
and implements the algorithm for finding witness computations described in~\cite{BDM13}.
Our implementation reuses that algorithm as a robustness-based oracle.
\trencher{} originally used \spin~\cite{Holzmann97} as back-end \seqcon\ reachability checker.
The current implementation, however, uses a simpler model checker that exploits information about the instruction set for partial-order reduction.
Moreover, it avoids having to compile the verifier executables (pan) as is the case for \spin.

We have implemented Algorithm~\ref{Algorithm:Reachability} with the following amendments.
First, the extension does not delete the store instruction $\aninstruction_1$.
This ensures the extended program has a (sound) superset of the \tso{} behaviors of the original program.
Second, the extension only adds instructions along $\extensionstate{1},\ldots,\extensionstate{n}$.
The remaining instructions were added to ensure all behaviors of the original program exist in the extended program, once $\aninstruction_1$ is removed.
The resulting algorithm is guaranteed to give correct results for cyclic programs. Of course, it cannot be guaranteed to terminate in general.
Finally, our implementation explores extensions due to different instruction sequences in parallel, rather than sequentially.

We compare our prototype implementation against two other model checkers that support \tso{} semantics:
\memorax~\cite{AbdullaACLR12} (revision 4f94ab6) and \cbmc~\cite{Clarke04atool} (version 4.7).
\memorax\ implements a sound and complete reachability checking procedure
by reducing to coverability in a well-structured transition system.
\cbmc{} is an SMT-based bounded model checker for C programs.
Consequently, it is sound, but not complete:
it is complete only up to a given bound on the number of loop iterations in the input program.
\subsection{Examples}
\begin{wrapfigure}{r}{0.625\textwidth}
	\vskip -2.75em
	\centering
	\footnotesize
	\setlength{\tabcolsep}{2.25pt}
\begin{tabular}{|c||l||r|r|r||r||r|r|}
\hline
\# & Program & T & St & Tr & RQ & CPU & Real\\
\hline
\hline
1&Parker (non-rob)&2&11&10&4&8&5\\
\hline
2&Peterson (non-rob)&2&14&18&12&21&13\\
\hline
3&Dekker (non-rob)&2&24&30&30&171&70\\
\hline
4&Lamport (non-rob)&3&33&36&27&1839&694\\
\hline
5&MCS Lock&4&52&50&30&127&61\\
\hline
6&CLH Lock&3&43&41&70&10&7\\
\hline
7&Lock-Free Stack&4&46&50&14&9&7\\
\hline
\end{tabular}
\caption{\trencher\ benchmarking results. The tests are available online~\cite{Trencher}. Times here are in milliseconds.}
	\vskip -2em
\label{Figure:Experiments}
\end{wrapfigure}

We tested our tool on a set of examples. %
Figure~\ref{Figure:Experiments} summarizes characteristics of the examples taken from the initial \trencher\ tests:
number of threads (T), states (St), and transitions (Tr).
The first example is a model of the buggy Parker class from Java VM~\cite{dice09:park}.
The next three ex\-am\-ples are mu\-tu\-al ex\-clu\-sion protocols implemented via shared va\-ri\-a\-bles.
These pro\-to\-cols do not guar\-an\-tee mu\-tu\-al ex\-clu\-sion un\-der \tso.
We tested Dekker's and Peterson's algorithms for two threads, and Lamport's fast mutex~\cite{lamport87:fast} for three threads.
The last three tests from Figure~\ref{Figure:Experiments} give statistics concerning reachability in robust test cases for the lock-free stack,
and for the MCS and CLH locking algorithms from \cite{Herlihy2008}.

We also performed three parametrized tests.
First, we varied the number of threads in Lamport's fast mutex~\cite{lamport87:fast} (see left-hand-side of Figure~\ref{Figure:Lamport}).
The modified Dekker in Figure~\ref{Figure:Diamond}
is inspired by the examples of the fence-insertion tool \textsc{musketeer}~\cite{AKNP14}
and adds an ``$N$-branching diamond'' (see right-hand-side of Figure~\ref{Figure:Diamond}) to both program threads.
Lastly, the program in Figure~\ref{Figure:Countdown} places stores to address $\xaddr$ on a length $N$ loop in thread $\athread_1$:
since $\athread_1$ expects to load the initial $\yaddr$ value
while $\athread_2$ expects to load $1$ and then $0$ from $\xaddr$,
an execution that reaches the goal state goes through the length $N$ loop twice.
\begin{figure}[!ht]
	\hspace{-0.25cm}
	\centering\small
	\begin{tikzpicture}[->,>=stealth',shorten >=1pt, auto, node distance=1cm, transform shape, scale=1]%
		\tikzstyle{every state}=[fill=white,circle,draw=black,inner sep=0pt,text=black,minimum size=6pt]
			\node[state,initial,initial text=$\athread_i$] (qstart) {};
			\node[state] (q0) [below of=qstart] {};
			\node[state] (q1) [below of=q0] {};
			\node[state] (q2) [below of=q1] {};
			\node[state] (q3) [below of=q2] {};
			\node[state] (q4) [below of=q3] {};
			\node[state] (q5) [node distance=3.75cm,right of=q2] {};
			\node[state] (q6) [above of=q5] {};
			\node[state] (q7) [above of=q6] {};
			\node[state,accepting] (qenter) [below of=q5] {};
			\node[state] (q9) [below of=qenter] {};
			\path	(qstart)	edge node [swap] {$\thelocal{\areg}{i}$} (q0)
					(q0)		edge node [swap] {$\thestore{\xaddr}{\areg}$} (q1)
					(q1)		edge node [swap] {$\theload{\areg_\yaddr}{\yaddr}$} (q2)
					(q2)		edge [bend right] node [swap] {$\theassume{\areg_\yaddr\hspace{-0.2pt}\neq\hspace{-0.2pt}0}$} (q0)
					(q2)		edge node [swap] {$\theassume{\areg_\yaddr\hspace{-0.2pt}=\hspace{-0.2pt}0}$} (q3)
					(q3)		edge node [swap] {$\thestore{\yaddr}{\areg}$} (q4)
					(q5)		edge node {$\theassume{\areg_\xaddr\hspace{-0.2pt}\neq\hspace{-0.2pt}\areg}$} (q6)
					(q6)		edge node {$\theload{\areg_\yaddr}{\yaddr}$} (q7)
					(q7)		edge node [swap] {$\theassume{\areg_\yaddr\hspace{-0.2pt}\neq\hspace{-0.2pt}\areg}$} (q0);
			\path	(q5)		edge node [swap] {$\theassume{\areg_\xaddr\hspace{-0.2pt}=\hspace{-0.2pt}\areg}$} (qenter)
					(qenter)	edge node [swap] {$\thestore{\yaddr}{0}$} (q9);
			\draw (q4.north east) .. controls +(60:1) and +(180:3.5) .. (q5.west) node [swap,pos=.45,xshift=-2pt] {$\theload{\areg_\xaddr}{\xaddr}$};
	\end{tikzpicture}
	\begin{tikzpicture}[transform shape,scale=.75,baseline=-1cm]
		\begin{axis}[xtick={2,3,4},ytick={0,23,39},xlabel={$N$ --- number of threads},ylabel={wall-clock time (seconds)},
			legend style={draw=none},legend pos={north west}]
			\addplot[smooth,color=green,mark=triangle] plot coordinates {
				(2,0.021)
				(3,0.878)
				(4,22.653)
			};
			\addplot[smooth,color=blue,mark=square] plot coordinates {
	        	(2,0.012)
				(3,0.694)
          	(4,38.817)
	    	};
			\addplot[smooth,color=red,mark=triangle] plot coordinates {
				(2,0.418)
				(3,0.643)
				(4,0.982)
			};
	    \legend{\memorax,\trencher,\cbmc}
	    \end{axis}
	\end{tikzpicture}
\caption{The $i$-th Lamport mutex thread (left) and running times for $N$ threads (right).}
\label{Figure:Lamport}
\end{figure}
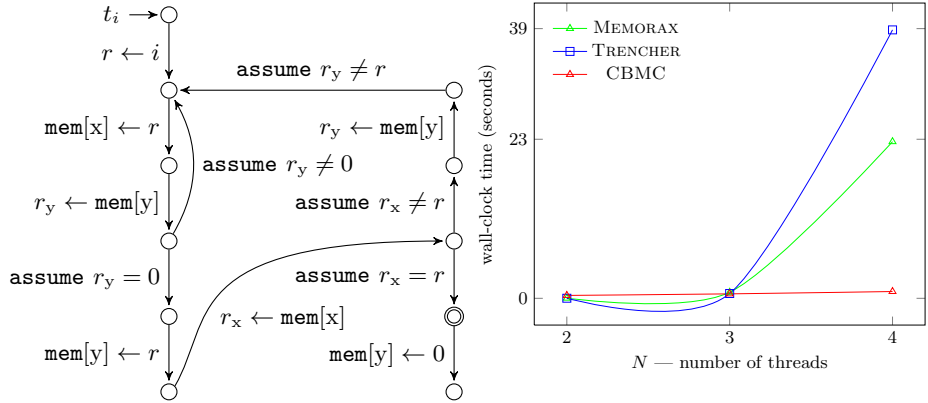

\begin{figure}[!ht]
  \vskip -1em
	\centering\small
	\begin{tikzpicture}[->,>=stealth',shorten >=1pt, auto, node distance=1cm, transform shape, scale=1]%
		\tikzstyle{every state}=[fill=white,circle,draw=black,inner sep=0pt,text=black,minimum size=6pt]
		\begin{scope} %
			\node[draw=none] (control) {}; %
			\node[state,initial,initial left,initial text=$\athread_1$] (q10) [node distance=0.5cm,above of=control] {};
			\node[state] (q11) [below of=q10] {};
			\node[state] (q1e) [below of=q11] {};
			\node[state] (q12) [below of=q1e] {};
			\node[accepting,state] (q13) [below of=q12] {};
			\path	(q10)		edge	node [swap] {$\thestore{\xaddr}{1}$} (q11)
					(q1e)		edge 	node [swap] {$\theload{r_1}{\yaddr}$} (q12)
					(q12)		edge	node [swap] {$\theassume{r_1\hspace{-.2em}=\hspace{-.1em}0}$} (q13);
			\path[dashed] (q11) edge node[swap] {$\Diamond_N(a)$} (q1e);
		\end{scope}
		\begin{scope} %
			\node[state,initial,initial left,initial text=$\athread_2$] (q20) [node distance=2.65cm,right of=q10] {};
			\node[state] (q21) [below of=q20] {};
			\node[state] (q2e) [below of=q21] {};
			\node[state] (q22) [below of=q2e] {};
			\node[accepting,state] (q23) [below of=q22] {};
			\path	(q20)		edge	node [swap] {$\thestore{\yaddr}{1}$} (q21)
					(q2e)		edge 	node [swap] {$\theload{r_2}{\xaddr}$} (q22)
					(q22)		edge	node [swap] {$\theassume{r_2\hspace{-.2em}=\hspace{-.1em}0}$} (q23);
			\path[dashed] (q21) edge node[swap] {$\Diamond_N(b)$} (q2e);
		\end{scope}
		\begin{scope} %
			\node[state,label=above:entry] (q0) [node distance=5.4cm, right of=control] {};
			\node[draw=none] (diamond) [left of=q0] {$\Diamond_N(a)$:};
			\node[state] (q1) [below of=q0] {};
			\node[draw=none] (empty) [below of=q1] {$\cdots$};
			\node[state] (q2) [node distance=0.5cm, left of=empty] {};
			\node[state] (q3) [node distance=0.5cm, right of=empty] {};
			\node[state,label=below:exit] (q4) [below of=empty] {};
			\path	(q0)		edge node {$\theload{\areg}{a}$} (q1)
					(q1)		edge node [yshift=-4pt,swap] {$\theassume{\areg\hspace{-0.2pt}=\hspace{-0.2pt}0}$} (q2)
					(q1)		edge node [yshift=-5pt] {$\theassume{\areg\hspace{-0.2pt}=\hspace{-0.2pt}i}$  $\forall i\in\intrange{1}{N-1}$} (q3)
					(q2)		edge node [yshift=6pt,swap] {$\thestore{a}{1}$} (q4)
					(q3)		edge node [yshift=6pt] {$\thestore{a}{(i+1)\hspace{-4pt}\mod{N}}$} (q4);	
		\end{scope}
	\end{tikzpicture}
\caption{Dekker's algorithm modified so that an ``$N$-branching diamond'' over distinct addresses $a,b\notin\set{\xaddr,\yaddr}$ 
is placed between the accesses to $\xaddr$ and $\yaddr$.
A final goal state is \tso{}-reachable if the first store is delayed past the last load in either $\athread_1$ or $\athread_2$.}
\label{Figure:Diamond}
\end{figure}
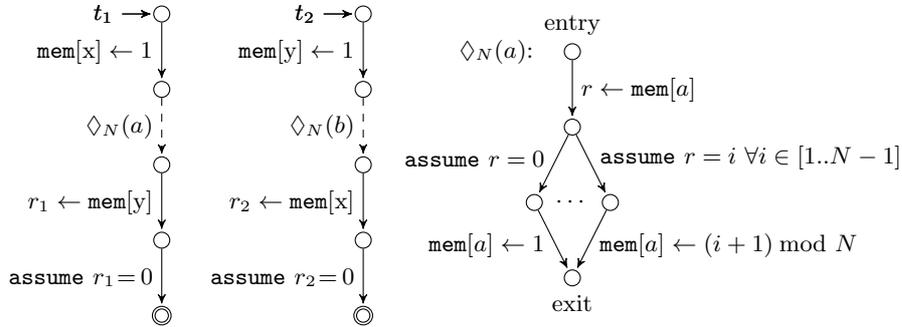

\begin{figure}[!ht]
	\hspace{-0.25cm}
	\centering\small
	\begin{tikzpicture}[->,>=stealth',shorten >=1pt, auto, node distance=1cm, transform shape, scale=1]%
		\tikzstyle{every state}=[fill=white,circle,draw=black,inner sep=0pt,text=black,minimum size=6pt]
		\begin{scope} %
			\node[state,initial,initial text=$\athread_1$] (q10) {};
			\node[state] (q11) [above of=q10] {};
			\node[state] (q12) [above of=q11] {};
			\node[state] (q13) [below of=q10] {};
			\node[state] (q14) [below of=q13] {};
			\node[accepting,state] (q15) [below of=q14] {};
			\path	(q10)		edge	node {$\theassume{\areg_1\hspace{-0.2pt}<\hspace{-0.2pt}N}$} (q11)
					(q10)		edge 	node [swap] {$\theassume{\areg_1\hspace{-0.2pt}=\hspace{-0.2pt}N}$} (q13)
					(q13)		edge	node [swap] {$\theload{\areg_2}{\yaddr}$} (q14)
					(q14)		edge	node [swap] {$\theassume{\areg_2\hspace{-0.2pt}=\hspace{-0.2pt}0}$} (q15)
					(q13)		edge [bend right]	node [swap] {$\thelocal{\areg_1}{0}$} (q10)
					(q11)		edge	node {$\thestore{\xaddr}{\areg_1}$} (q12)
					(q12)		edge [bend left] node [midway] {$\thelocal{\areg_1}{\areg_1+1}$} (q10);
		\end{scope}
		\begin{scope}[bend angle=20] %
			\node[state,initial,initial text=$\athread_2$,node distance=3.65cm] (q20) [right of=q12] {};
			\node[state] (q21) [below of=q20] {};
			\node[state] (q22) [below of=q21] {};
			\node[state] (q23) [below of=q22] {};
			\node[state] (q24) [below of=q23] {};
			\node[accepting,state] (q25) [below of=q24] {};
			\path	(q20)		edge	node [swap] {$\thestore{\yaddr}{1}$} (q21)
					(q21)		edge 	node [swap] {$\theload{\areg_3}{\xaddr}$} (q22)
					(q22)		edge	node [swap] {$\theassume{\areg_3\hspace{-0.2pt}=\hspace{-0.2pt}1}$} (q23)
					(q23)		edge 	node [swap] {$\theload{\areg_3}{\xaddr}$} (q24)
					(q24)		edge 	node [swap] {$\theassume{\areg_3\hspace{-0.2pt}=\hspace{-0.2pt}0}$} (q25);
		\end{scope}	
	\end{tikzpicture}
	\begin{tikzpicture}[transform shape,scale=.75,baseline=-1cm]
		\begin{axis}[xtick={3,5,7,9,11,13},ytick={0,8,15,45},xlabel={$N$ --- length of the loop},ylabel={wall-clock time (seconds)},legend style={draw=none},legend pos={north west}]
			\addplot[smooth,blue,mark=square] plot coordinates {
	        	(3,0.008)
          	(5,0.009)
		     	(7,0.013)
		     	(9,0.020)
	         	(11,0.026)
				(13,0.034)
	    	};
	    	\addplot[smooth,color=red,mark=triangle] plot coordinates {
	         	(3,0.452)
	         	(5,0.939)
	         	(7,2.263)
	         	(9,8.037)
	         	(11,15.156)
	         	(13,45.963)
	     	};
	    \legend{\trencher,\cbmc}
	    \end{axis}
	\end{tikzpicture}
\caption{A final goal state is \tso-reachable if $\athread_1$ goes through the (length $N$) loop two times:
once to satisfy $\theassume{\areg_3=1}$ and the second time to satisfy $\theassume{\areg_3=0}$.}
\label{Figure:Countdown}
\end{figure}
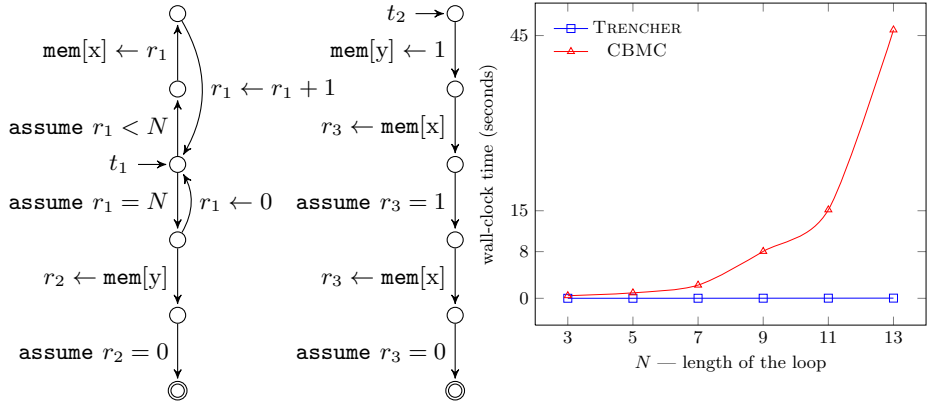

\subsection{Evaluation}
We ran all tests on a QEMU @ 2.67GHz virtual machine (16 cores) with 8GB RAM run\-ning GNU/Linux.
The table in Figure~\ref{Figure:Experiments} summarizes the results of the \trencher\ benchmark tests.
RQ is the number of \seqcon{} reachability queries raised by \trencher.
The columns CPU and Real give the total CPU time and the wall-clock time for performing a test.

The first graph in Figure~\ref{Figure:ResultsGraphs} depicts the running times of the three tools on the non-robust examples
from Figure~\ref{Figure:Experiments}.
For \cbmc, we used the versions of the mutual exclusion algorithms that its authors provide.
For \memorax, we hand-wrote {\tt *.rmm} files for the first 4 test programs.
We did not perform a comparison for robust programs:
if \seqcon{} reachability returns false on an input program,
our implementation decides mutual exclusion as fast as \trencher\ is able to determine robustness.
Moreover, \cbmc\ implements strictly an under-approximative method where the number of loop iterations is bounded.
Our robust tests, however, contain unbounded loops.

\begin{figure}[!t]
	\centering%
	\footnotesize
	\begin{tikzpicture}[transform shape,scale=.75,baseline=1.95cm]
		\begin{axis}[xlabel={Test order index from Figure~\ref{Figure:Experiments}.},ylabel={wall-clock time (seconds)},legend style={draw=none},
      legend pos={north west},xtick={1,2,3,4},ybar,bar width=7pt, xtick align=inside]
	    	\addplot[color=green!50!black, pattern = north west lines] plot coordinates { %
	       	(1,0.012) %
	       	(2,0.111) %
    		  (3,0.028) %
  		   	(4,0.831) %
			};
			\addplot[color=blue, pattern = horizontal lines] plot coordinates { %
	      	(1,0.005)
	       	(2,0.013)
	  			(3,0.07)
		  		(4,0.694)
	    	};
	    	\addplot[color=red, pattern = north east lines] plot coordinates { %
	       	(1,0.451)
			  	(2,0.329)
				  (3,0.345)
				  (4,0.643)
	     	};
			\addplot [sharp plot, update limits=false] coordinates {(0,-.0025) (6,-.0025)};
	    \legend{\memorax,\trencher,\cbmc}
	    \end{axis}
	\end{tikzpicture}
		\begin{tikzpicture}[transform shape,scale=.75,baseline=1.95cm]
		\begin{axis}[xlabel={$N$ --- diamond branching factor},ylabel={wall-clock time (seconds)},legend style={draw=none},legend pos={north west},
			xtick={10,20,30,40},ytick={0,15,30,60}]
			\addplot[smooth,blue,mark=square] plot coordinates {
				(10,0.008)
				(15,0.010)
				(20,0.012)
				(25,0.016)
				(30,0.025)
				(35,0.031)
				(40,0.038)
	    	};
	    	\addplot[smooth,color=green!50!black,mark=triangle] plot coordinates {
				(10,0.958)
				(15,2.932)
				(20,6.881)
				(25,13.820)
				(30,23.996)
				(35,40.350)
				(40,63.498)
	     	};
	    \legend{\trencher,\memorax}
	    \end{axis}
	\end{tikzpicture}
\caption{Running times for the non-robust tests in Figure~\ref{Figure:Experiments} (left) and Figure~\ref{Figure:Diamond} (right).}
\label{Figure:ResultsGraphs}
\end{figure}
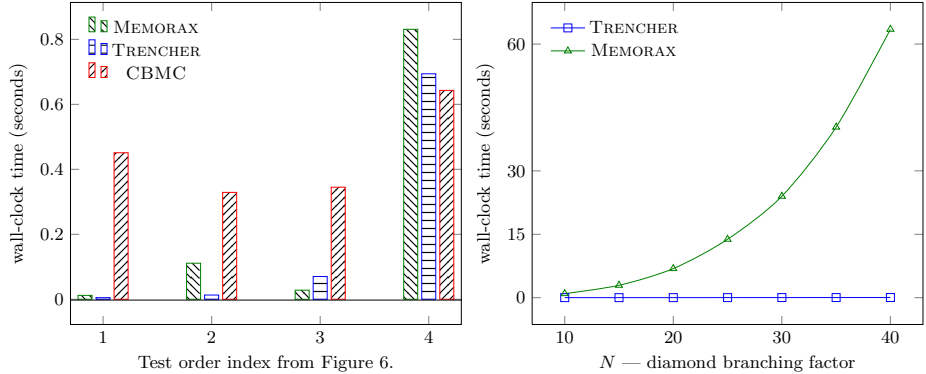

The high load needed to verify Lamport's mutex --- in comparison with the other Figure~\ref{Figure:Experiments} tests ---
is justified by the correlation between the program's data domain size and its number of threads.
For a larger number of threads, the right-hand-side graph in Figure~\ref{Figure:Lamport} shows that \cbmc\ is fastest. 
This is the case since, actually, the smallest unwind bound suffices for \cbmc\ to conclude reachability.
For \memorax\ and \trencher\ the system runs out of memory when $N=5$.
This underlines once again just how troublesome the state-space explosion is for \tso\ reachability. 
Although it is not easily noticeable in the picture, \memorax's exponential scaling is better than \trencher's: 
although \trencher\ is slightly faster than \memorax\ for $N\in\set{2,3}$, \memorax\ clearly outperforms \trencher\ when $N=4$.

The graph in Figures~\ref{Figure:Countdown} show that,
for the second parameterized test,
our prototype is faster than \cbmc.
Indeed, with increasing $N$, an ever larger number of constraints need to be generated by \cbmc.
For \trencher{}, regardless of the value of $N$, it takes three \seqcon{} reachability queries to conclude \tso{} reachability.

The second graph in Fig\-ure~\ref{Figure:ResultsGraphs} shows that,
for the programs described by Figure~\ref{Figure:Diamond},
our prototype is faster than \memorax.
It seems \memorax\ cannot cope well with the branching factor that the parameter $N$ introduces.

To better understand the difficulty of the latter two parametric tests,
we present the exponential scaling behaviors of \trencher\ in Figure~\ref{Figure:AppendixGraphs}.
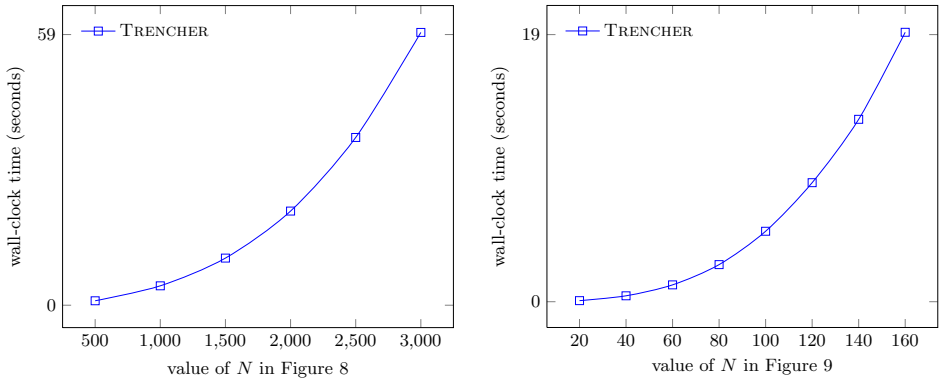
\begin{figure}[!h]
	\hspace{-0.25cm}
	\centering\small
	\begin{tikzpicture}[transform shape,scale=.75,baseline=-1cm]
		\begin{axis}[xlabel={value of $N$ in Figure~\ref{Figure:Diamond}},ylabel={wall-clock time (seconds)},legend style={draw=none},legend pos={north west},
			xtick={500,1000,1500,2000,2500,3000},ytick={0,59}]
			\addplot[smooth,blue,mark=square] plot coordinates {
				(500, 0.968)
	        	(1000, 4.228)
		     	(1500, 10.272)
				(2000, 20.523)
		     	(2500, 36.567)
		     	(3000, 59.465)
	    	};
	    \legend{\trencher}
	  \end{axis}
	\end{tikzpicture}
	\hfill
	\begin{tikzpicture}[transform shape,scale=.75,baseline=-1cm]
		\begin{axis}[xlabel={value of $N$ in Figure~\ref{Figure:Countdown}},ylabel={wall-clock time (seconds)},legend style={draw=none},legend pos={north west},
			xtick={20,40,60,80,100,120,140,160},ytick={0,19}]
			\addplot[smooth,blue,mark=square] plot coordinates {
	         	(20, 0.076)
				(40, 0.421)
				(60, 1.199)
				(80, 2.639)
				(100, 5.012)
				(120, 8.467)
				(140, 12.974)
				(160, 19.164)
	    	};
	    \legend{\trencher}
	  \end{axis}
	\end{tikzpicture}
\caption{Additional \trencher\ results for the programs in Figures~\ref{Figure:Diamond} and~\ref{Figure:Countdown}.
\memorax\ takes already 1 minute and 24 seconds for the program in Figure~\ref{Figure:Diamond} and $N=50$,
while \cbmc\ takes 8 minutes and 35 seconds for the program in Figure~\ref{Figure:Countdown} and $N=20$.}
\label{Figure:AppendixGraphs}
\end{figure}

\subsection{Discussion}
Because we find several witnesses in parallel, throughout the experiments our implementation required up to $2$ iterations of the loop in Algorithm~\ref{Algorithm:Reachability}.
In the case of robust programs, one iteration is always sufficient.
This suggests that robustness violations are really the critical behaviors leading to \tso\ reachability.

The experiments indicate that, at least for some programs with a high branching factor,
our implementation is faster than \memorax\ if a useful witness can be found within a small number of iterations of Algorithm~\ref{Algorithm:Reachability}.
Similarly, our prototype is better than \cbmc\ for programs which require a high unwinding bound to make visible \tso{} behavior
reaching a goal state.
Although the two programs by which we show this are rather artificial,
we expect such characteristics to occur in actual code. 
Hence, our approach seems to be strong on an orthogonal set of programs.
In a portfolio model checker, it could be used as a promising alternative to the existing techniques.

To evaluate the practicality of our method, more experiments are needed.
In particular, we hope to be able to substantiate the above conjecture for concrete programs with behavior like that depicted in Figures~\ref{Figure:Diamond} and~\ref{Figure:Countdown}.
Unfortunately, there seems to be no clear way of translating (compiled) C programs into our simplified assembly syntax without substantial abstraction.
To handle C code, an alternative would be to reimplement our method within \cbmc.
But this would force us to determine a-priori a good-enough unwinding bound.
Moreover, we could no longer conclude safety of robust programs with unbounded loops.
\acks
The third author was granted by the Competence Center High Performance Computing and Visualization (CC-HPC)
of the Fraunhofer Institute for Industrial Mathematics (ITWM).
The work was partially supported
by the PROCOPE project ROIS: \emph{Robustness under Realistic Instruction Sets} and
by the DFG project R2M2: \emph{Robustness against Relaxed Memory Models}.

\begin{spacing}{0.9}
\bibliographystyle{plain}%

\end{spacing}
\clearpage
\appendix\appendixtrue
\section{A Simple Safe Program}\label{Appendix:SafeNonterminating}
The program from Figure~\ref{Figure:NonTerminating} is safe since no goal state is \tso{}-reachable:
the initial control states will never be left since the conditionals will never succeed.  
However, the algorithm that we describe for Theorem~\ref{Theorem:PartialCorrectness} does not terminate for this example. 
Although every $\aprogram_k$ that unrolls the program in Figure~\ref{Figure:NonTerminating} up to $k\in\nat$ is found safe,
the algorithm only stops if a \tso{}-reachable state is found or if $\oracleOf{\bprog}=\epsilon$, which is never the case.
\begin{figure}[!ht]
	\vskip -0.5em
	\centering\small
	\begin{tikzpicture}[->,>=stealth',shorten >=1pt, auto, node distance=3cm, transform shape, scale=1,baseline=.65cm]%
		\tikzstyle{every state}=[fill=white,draw=black,text=black,minimum size=.1pt]
		\begin{scope} %
			\node[state,initial,initial text=$\athread_1$,label=below right:$\acontrolstate_{0,1}$] (q10) {};
			\node[state,accepting,label=below right:$\acontrolstate_{f,1}$] (q11) [right of=q10] {};
			\path (q10)		edge node {$\theassume{\areg_1=2}$} (q11)
								edge [loop above] node {$\thestore{\xaddr}{1-\areg_1}$} ()
								edge [loop below] node {$\theload{\areg_1}{\yaddr}$} ();
		\end{scope}
		\begin{scope} %
			\node[state,initial,initial text=$\athread_2$,label=below right:$\acontrolstate_{0,2}$] (q20) [node distance=3cm,right of=q11] {};
			\node[state,accepting,label=below right:$\acontrolstate_{f,2}$] (q21) [right of=q20] {};
			\path (q20)		edge node {$\theassume{r_2=2}$} (q21)
								edge [loop above] node {$\thestore{\yaddr}{1-\areg_2}$} ()
								edge [loop below] node {$\theload{\areg_2}{\xaddr}$} ();		
		\end{scope}	
	\end{tikzpicture}
	\caption{A safe program for which Algorithm~\ref{Algorithm:Reachability} (as in Theorem~\ref{Theorem:PartialCorrectness}) does not terminate.}
	\label{Figure:NonTerminating}
	\vskip -0.5em
\end{figure}
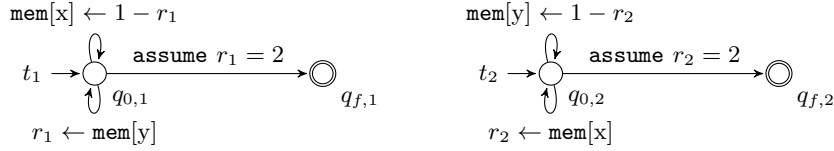

The underlying reason why always $\oracleOf{\bprog}\neq\epsilon$ is that there are infinitely many sequences 
$\aninstruction_\text{store}^m\cdot\aninstruction_\text{load}$, where  
$\aninstruction_\text{store}=\atransition{\acontrolstate_{0,1}}{\acontrolstate_{0,1}}{\thestore{\xaddr}{1-\areg_1}}$,
$\aninstruction_\text{load}=\atransition{\acontrolstate_{0,1}}{\acontrolstate_{0,1}}{\theload{\areg_1}{\yaddr}}$,
and $m\in\nat$.

\section{Proofs missing in Subsection~\ref{Subsection:SoundAndComplete}}\label{Appendix:MissingProofs}
\printproofs
\end{document}